\documentclass[journal,10pt]{IEEEtran}

\PassOptionsToPackage{numbers,sort&compress}{natbib}

\usepackage{xcolor}
\usepackage{graphicx}
\usepackage[utf8]{inputenc} 
\usepackage[T1]{fontenc}    
\usepackage{hyperref}       
\usepackage{url}            
\usepackage{booktabs}       
\usepackage{amsfonts}       
\usepackage{nicefrac}       
\usepackage{microtype}      
\usepackage{amsmath}
\usepackage{amssymb}
\usepackage{natbib}
\usepackage[capitalize]{cleveref}
\Crefname{figure}{Fig.}{Figs.}
\crefname{equation}{}{}
\Crefname{equation}{Equation}{Equations}
\usepackage{amsthm}
\newtheorem{theorem}{Theorem}
\newtheorem{proposition}{Proposition}
\newtheorem{corollary}{Corollary}
\newtheorem{lemma}{Lemma}
\theoremstyle{definition}
\newtheorem{definition}{Definition}
\usepackage{bm}
\usepackage{subcaption}
\usepackage{autonum}
\usepackage{unicmd}
\graphicspath{{../src/}}
\makeatletter
\DeclareRobustCommand{\idparen@footnote}[1]{\raisebox{0.93pt}{\scalebox{.7}{$\scriptstyle($}}#1\raisebox{0.93pt}{\scalebox{.7}{$\scriptstyle)$}}}
\def\uidfootnote#1#2{#1^{\idparen@footnote{#2}}}
\def\vhat#1{\decosmash\hat{#1}}
\def\vtilde#1{\decosmash\tilde{#1}}
\def\vadjoint#1{\decosmash\adjoint{#1}}

\let\ksum\oplus
\def\cprod{\mathbin{\square}}
\def\vtrsps#1{\vphantom{#1}\smash{\trsps{#1}}}

\newcommand{\IN}{\mathrm{in}}
\newcommand{\OUT}{\mathrm{out}}

\newcommand{\opt}{\mathrm{opt}}
\newcommand{\newword}[1]{#1}
\makeatother

\title{Multi-dimensional Graph Fourier Transform}

\author{%
  Takashi~Kurokawa,~\IEEEmembership{Non-Member,~IEEE,} Taihei~Oki,~\IEEEmembership{Non-Member,~IEEE,} and~Hiromichi~Nagao,~\IEEEmembership{Non-Member,~IEEE}%
  \thanks{%
    Manuscript received ; revised .
    This work was supported by JSPS KAKENHI Grants-in-Aid for Scientific Research (B) (Grant Numbers JP26280006 and JP17H01703).

    T. Kurokawa and T. Oki are with Department of Mathematical Informatics, Graduate School of Information Science and Technology, The University of Tokyo, Tokyo 113-8656, Japan (e-mail: takashi\_kurokawa@mist.i.u-tokyo.ac.jp, taihei\_oki@mist.i.u-tokyo.ac.jp).
    H. Nagao is with Earthquake Research Institute, The University of Tokyo, Tokyo 113-0032, Japan (e-mail: nagaoh@eri.u-tokyo.ac.jp), and with Department of Mathematical Informatics, Graduate School of Information Science and Technology, The University of Tokyo, Tokyo 113-8656, Japan.}
  }

\begin{document}

\maketitle

\begin{abstract}
  Many signals on Cartesian product graphs appear in the real world, such as digital images, sensor observation time series, and movie ratings on Netflix.
  These signals are ``multi-dimensional'' and have directional characteristics along each factor graph.
  However, the existing graph Fourier transform does not distinguish these directions, and assigns 1-D spectra to signals on product graphs.
  Further, these spectra are often multi-valued at some frequencies.
  Our main result is a multi-dimensional graph Fourier transform that solves such problems associated with the conventional GFT.
  Using algebraic properties of Cartesian products, the proposed transform rearranges 1-D spectra obtained by the conventional GFT into the multi-dimensional frequency domain, of which each dimension represents a directional frequency along each factor graph.
  Thus, the multi-dimensional graph Fourier transform enables directional frequency analysis, in addition to frequency analysis with the conventional GFT.
  Moreover, this rearrangement resolves the multi-valuedness of spectra in some cases.
  The multi-dimensional graph Fourier transform is a foundation of novel filterings and stationarities that utilize dimensional information of graph signals, which are also discussed in this study.
  The proposed methods are applicable to a wide variety of data that can be regarded as signals on Cartesian product graphs.
  This study also notes that multivariate graph signals can be regarded as 2-D univariate graph signals.
  This correspondence provides natural definitions of the multivariate graph Fourier transform and the multivariate stationarity based on their 2-D univariate versions.
\end{abstract}

\section{Introduction}\label{sec:introduction}
Signals located on the vertices of weighted graphs, known as graph signals, appear in many situations, e.g., measurements on sensor network graphs, electrical potential on neural network graphs, and RGB on pixels on grid graphs.
Recently, many graph signal processing (GSP) methodologies have been proposed for graph signals, e.g., graph Fourier transform (GFT)~\citep{Taubin1995,Shuman2013,Shuman2016,Loukas2016,Hammond2011,Agaskar2013,Sandryhaila2014a,Sandryhaila2014}, windowed GFT~\citep{Shuman2016}, graph wavelet transform~\citep{Hammond2011}, and spectral filtering~\citep{Zhang2008,Grady2010}.
These GSP methodologies extended conventional signal processing techniques applicable to time-series data.

This study focuses on signals on a Cartesian product of graphs, which are termed ``multi-dimensional graph signals'' hereafter.
In summary, a Cartesian product of $n$ graphs is an ``$n$-dimensional graph'' whose each dimension is formed by each factor graph (its definition will be introduced in \cref{subsec:cpgraph}).
\Cref{fig:signal_intro} shows an example of a two-dimensional (2-D) graph signal.
The graph seems to have two dimensions: a horizontal graph of $5$ vertices (solid lines) and a vertical path of $4$ vertices (dotted lines).
We mention three practical 2-D graph signals.
First, images are signals on a grid graph that is a product of a row and column path graph.
Second, sensor observation time series are signals on a product of a path graph (model of time axis) and a sensor network graph (model of spatial correlation).
Third, Netflix movie ratings are signals on a product of a movie-similarity graph and a user-similarity graph.

\begin{figure}
  \captionsetup[subfigure]{justification=centering}
  \centering
  \begin{subfigure}[t]{0.7\linewidth}
    \centering
    \includegraphics[width=\columnwidth]{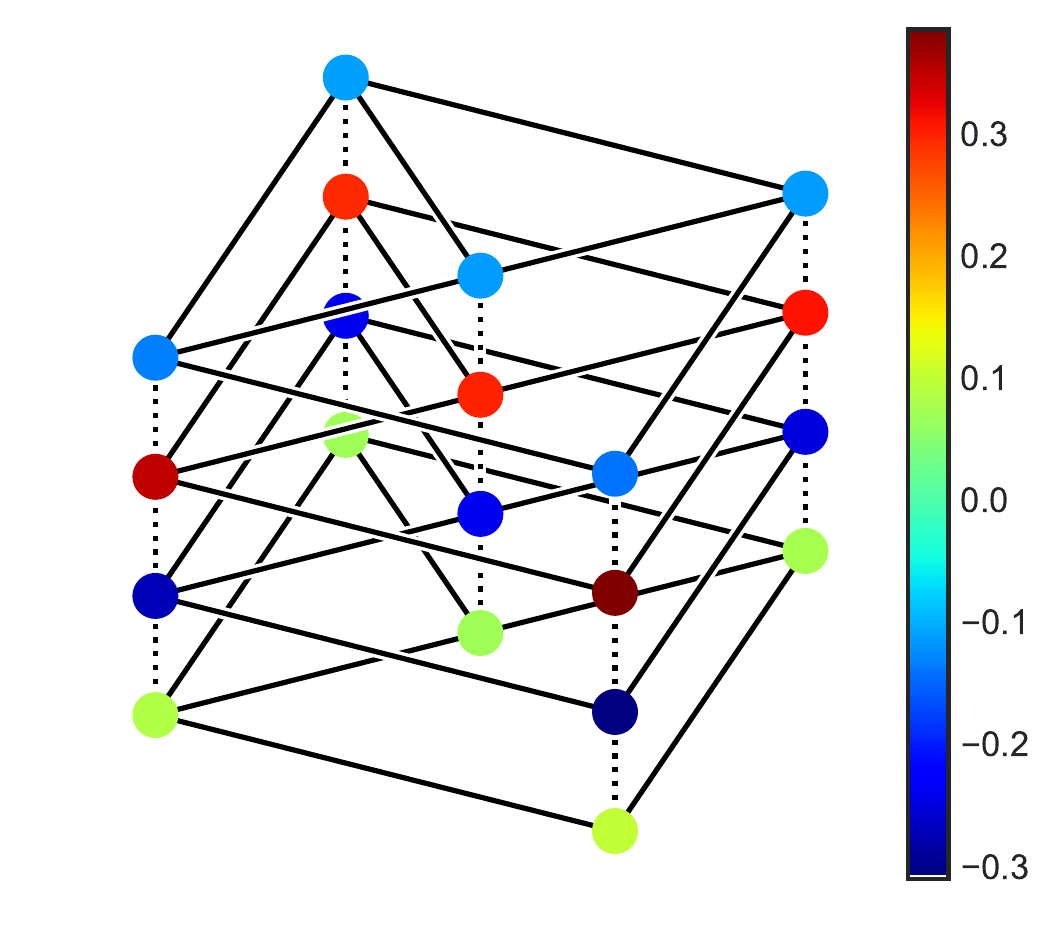}
    \caption{\vphantom{2-D graph signal}}
    \label{fig:signal_intro}
  \end{subfigure}
  \begin{subfigure}[t]{0.47\linewidth}
    \centering
    \includegraphics[width=\columnwidth]{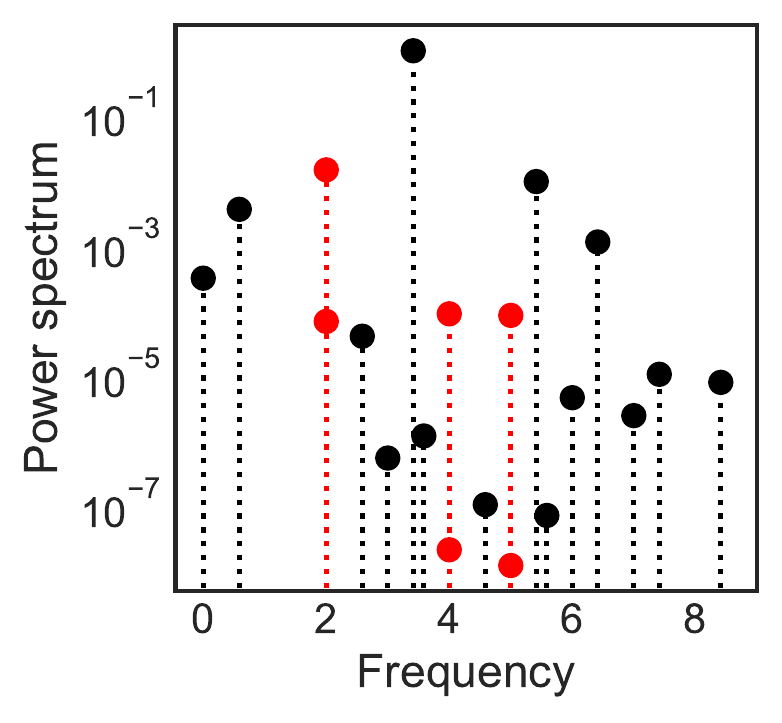}
    \caption{\vphantom{Power spectrum of the signal shown in \cref{fig:signal_intro} obtained by a GFT}}
    \label{fig:1dspec_intro}
  \end{subfigure}
  \begin{subfigure}[t]{0.505\linewidth}
    \centering
    \includegraphics[width=\columnwidth]{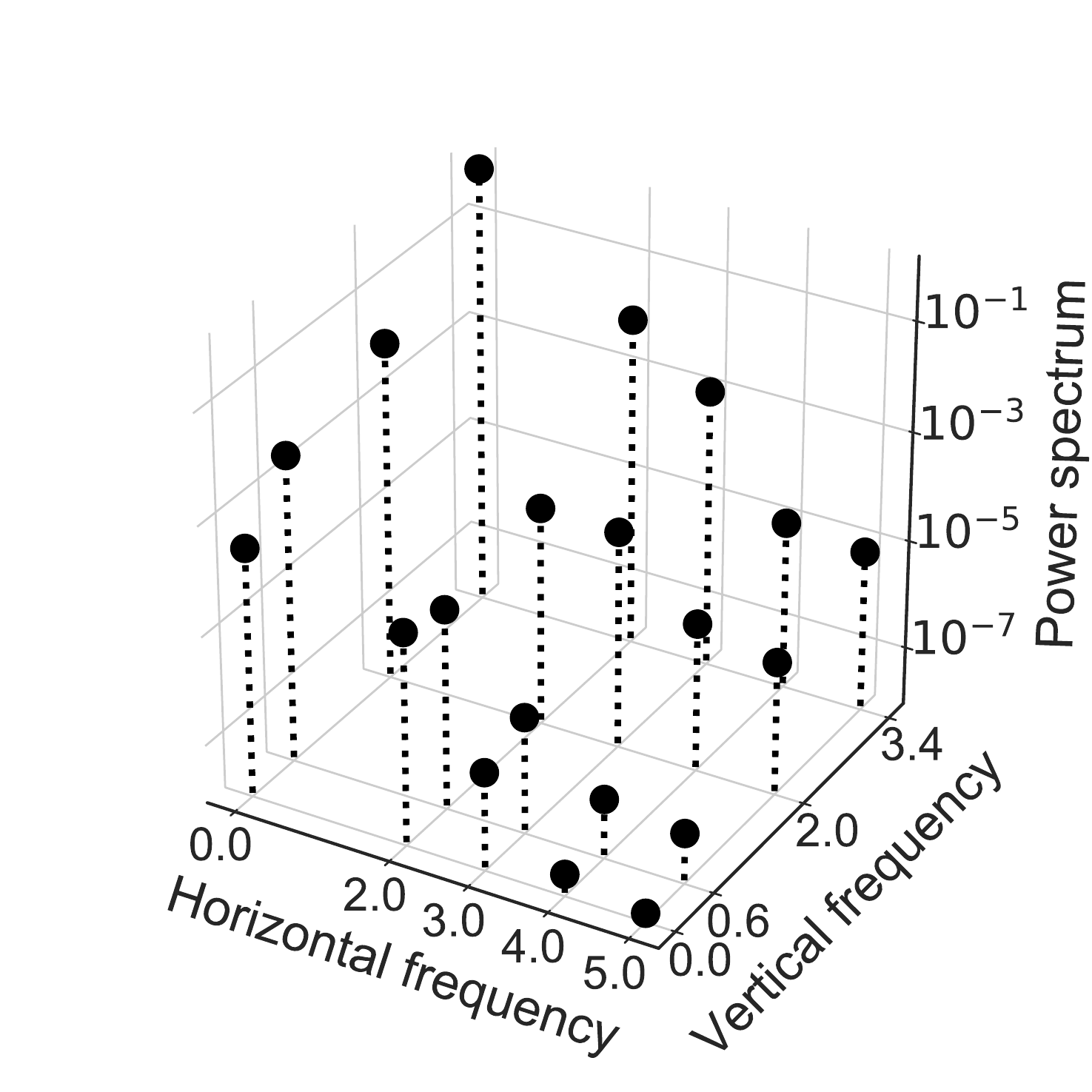}
    \caption{\vphantom{Power spectrum of the signal shown in \cref{fig:signal_intro} obtained by a 2-D GFT}}
    \label{fig:2dspec_intro}
  \end{subfigure}
  \caption{\subref{fig:signal_intro} Example of 2-D graph signal. The intensity of the signal is indicated by color on each vertex. \subref{fig:1dspec_intro} 1-D power spectrum of \subref{fig:signal_intro} obtained by a conventional graph Fourier transform (GFT), and \subref{fig:2dspec_intro} 2-D power spectrum of \subref{fig:signal_intro} obtained by the proposed multi-dimensional graph Fourier transform (MGFT). The 1-D spectrum is double-valued at red points, whereas the 2-D spectrum is single-valued everywhere.}
  \label{fig:signal_and_spec_intro}
\end{figure}

However, the conventional GFT is inappropriate for multi-dimensional graph signals.
First, the spectra of multi-dimensional graph signals obtained by the GFT cannot represent directional frequency characteristics, which makes an investigation of the original signals inadequate.
Signals on a product graph have ``directional'' characteristics along each factor graph, but the GFT maps the multi-dimensional signals to 1-D graph spectra, ignoring the directional information.
The disadvantage of the conventional GFT can be easily confirmed in \cref{fig:1dspec_intro}.
The power spectrum loses directional frequency characteristics, i.e., the original signal varies moderately in the horizontal direction and drastically in the vertical direction.
Second, the spectra of multi-dimensional graph signals obtained by the GFT tend to be multi-valued at a particular frequency.
The GFT spectra are functions that range from discrete frequencies to complex numbers.
When the Laplacian matrix of a given graph has non-distinct eigenvalues, the spectrum functions are not well-defined at a frequency of multiple eigenvalues, i.e., it associates many values to the frequency.
Unfortunately, regarding Cartesian product graphs, the eigenvalues of a Laplacian matrix frequently degenerate.
For example, the Laplacian matrix of the Cartesian product of two isomorphic graphs has non-distinct eigenvalues.
Indeed, the graph in \cref{fig:signal_intro} has multiple eigenvalues, $2$, $4$, and $5$, so that the spectrum in \cref{fig:1dspec_intro} is double-valued at these frequencies.

This study proposes a multi-dimensional graph Fourier transform (MGFT) for multi-dimensional graph signals that solves the aforementioned problems associated with the conventional GFT.
When given an $n$-D graph signal, the MGFT rearranges the 1-D spectrum obtained by the GFT into the $n$-D frequency domain, and provides the $n$-D spectrum of the signal.
Each dimension of the frequency space indicates a ``directional frequency along each factor graph.''
In this manner, the frequency characteristics acquired by the proposed MGFT include directional information, in addition to the frequency characteristics provided by the conventional GFT.
Therefore, the MGFT provides deeper frequency analysis of the multidimensional graph signals than the conventional GFT.
Moreover, the multi-dimensional rearrangement of 1-D spectra resolves the multi-valuedness of the spectra and generates well-defined spectrum functions under the conditions mentioned in \cref{subsec:mgft}.
\Cref{fig:2dspec_intro} shows the power spectrum of the signal in \cref{fig:signal_intro} obtained by the MGFT.
The spectrum in \cref{fig:2dspec_intro} is a 2-D rearrangement of the 1-D spectrum in \cref{fig:1dspec_intro}, indicates the anisotropic frequency characteristics of the original signal, and is single-valued at any frequency.
In addition, the MGFT is as fast as an efficient GFT algorithm on Cartesian product graphs~\citep{Sandryhaila2014}.

Aside from the significance of the MGFT itself, the MGFT is the foundation of the various GSP tools provided in this study that utilize dimensional information of graph signals.
We propose multi-dimensional graph spectral and optimization filtering, with which we can design directional frequency characteristics of multi-dimensional graph signals.
We also propose a factor-graph-wise stationarity and a directional stationarity of multi-dimensional graph signals that focus on stationarities along factor graphs. Further, we discuss the mutual relationships of the proposed stationarities and an existing stationarity in~\citep{Segarra2017,Perraudin2017}.

This study points out that multivariate signals on graphs can be regarded as signals on Cartesian product graphs.
Therefore, our study of multi-dimensional graph signals transfers to those of multivariate graph signals.
In this manner, we propose the stationarity of multivariate graph signals.

There have been several previous studies conducted on GSP methodologies for signals on a Cartesian product graph.
\Citet{Sandryhaila2014} proposed an efficient algorithm of an adjacency-based GFT when it is applied to signals on a product graph.
The algorithm reduces the high computational cost by utilizing an algebraic property of Cartesian products, which is similar to the proposed MGFT.
However, the two methods are essentially different in the following manner: their GFT eventually provides 1-D spectra, whereas the MGFT provides multi-dimensional spectra.
Other studies~\citep{Loukas2016,Loukas2017} modeled the periodic time axis using a cycle graph, and indicated that periodic temporal signals on a graph could be regarded as data on the Cartesian product of the graph and the cycle graph.
They proposed a set of methodologies consisting of ``joint graph and temporal Fourier transform,'' joint filtering, and some novel stationarities for such signals.
The set of methodologies consisting of the proposed multi-dimensional GFT, filtering, and stationarities generalizes their methodologies.

The contents of this study are described as follows.
\Cref{sec:gsp_framework} overviews two key ingredients of an MGFT: a conventional GFT and a Cartesian product of graphs.
\Cref{sec:mgft} proposes the MGFT.
The MGFT motivates several new filterings in \cref{sec:filtering} and stationarities in \cref{sec:stationarity} for multi-dimensional graph signals.
\Cref{sec:multivariate_gsp} discusses multivariate graph signals related to multi-dimensional graph signals.
\Cref{sec:conclusion} concludes the paper and mentions several future works.

\section{Preliminaries}\label{sec:gsp_framework}

\subsection{Notations}\label{subsec:notations}
Let $\setR$ be the set of all real numbers, $\setRnonneg$ be the set of all nonnegative real numbers, and $\setC$ be the set of all complex numbers.
A complex conjugate of a scalar $a$ is denoted by $\wideconj{a}$, an element-wise complex conjugate of a vector $\vec{a}$ is denoted by $\wideconj{\vec{a}}$, and that of a matrix $\mat{A}$ is denoted by $\wideconj{\mat{A}}$.
For a matrix $\mat{A}$, let $\trsps{\mat{A}}$ be a transpose, $\adjoint{\mat{A}}$ be a Hermitian transpose, $\mat{A}^{-1}$ be an inverse, and $\tr\mat{A}$ be a trace of $\mat{A}$.
A Kronecker product of $\mat{A}\in\setMatrix{m}{n}{\setC}$ and $\mat{B}\in\setMatrix{s}{t}{\setC}$ is a matrix given by
\begin{equation}
  \mat{A}\kprod\mat{B}
  =
  \begin{pmatrix}
    a_{11}\mat{B} & \cdots & a_{1n}\mat{B}
    \\
    \vdots & \ddots & \vdots
    \\
    a_{m1}\mat{B} & \cdots & a_{mn}\mat{B}
  \end{pmatrix}
  \in\setMatrix{ms}{nt}{\setC},
\end{equation}
where $a_{ij}$ is the $\seqprn{i,j}$-th element of $\mat{A}$.
For $p>0$, a $p$-norm of a univariate function $f$ on some discrete domain is
\begin{align}
  \norm[p]{f}=\paren{\sum_i\absprn{\app{f}{i}}^p}^{1/p},\label{eq:unifunction_norm}
\end{align}
and a $p$-norm of a bivariate function $g$ on some discrete domain is
\begin{align}
  \norm[p]{g}=\paren{\sum_i\sum_j\absprn{\app{g}{i,j}}^p}^{1/p}.
\end{align}
For a bivariate function $g$, let $\app{g}{i,\placeholder}$ be a univariate version of $g$ with the first variable fixed as $i$ and $\app{g}{\placeholder,j}$ a univariate version of $g$ with the second variable fixed as $j$.

\subsection{Graph Fourier transform}\label{subsec:gft}

Let $G=\seqprn{V,E,w}$ be \newword{an undirected weighted graph} with vertex set $V=\setprn{0,\ldots,N-1}$, edge set $E$, and weight function $\funcdoms{w}{V\times V}{\setRnonneg}$.
The weight function $w$ is symmetric and satisfies $\app{w}{i,j}=0$ for any $\setprn{i,j}\notin E$.
In this study, we assume that all graphs are simple, i.e., have no loops and no multiple edges.

The following three matrices associated with $G$ are significant: an adjacency matrix $\mat{W}=\rawparen{\app{w}{i,j}}_{i,j=0,\ldots,N-1}$, a degree matrix $\mat{D}$ whose $i$-th diagonal element is $\app{d}{i}=\sum_{j=0}^{N-1}\app{w}{i,j}$, and a Laplacian matrix $\mat{L}=\mat{D}-\mat{W}$.
In particular, the Laplacian matrix is necessary for GFTs.

Because the Laplacian matrix $\mat{L}$ is real, symmetric, and positive-semidefinite, it has nonnegative eigenvalues $\lambda_0,\ldots,\lambda_{N-1}$ and the corresponding orthonormal eigenfunctions $\funcdoms{u_0,\ldots,u_{N-1}}{V}{\setC}$ satisfying
\begin{align}
  \mat{L}
  \begin{pmatrix}
    \app{u_k}{0}
    \\
    \vdots
    \\
    \app{u_k}{N-1}
  \end{pmatrix}
  =\lambda_k
  \begin{pmatrix}
    \app{u_k}{0}
    \\
    \vdots
    \\
    \app{u_k}{N-1}
  \end{pmatrix}
\end{align}
for $k=0,\ldots,N-1$.
Here, the orthonormality means that $\sum_{i=0}^{N-1}\app{u_k}{i}\wideconj{\app{u_l}{i}}=\app{\delta}{i,j}$ holds for any $k,l=0,\ldots,N-1$, where $\delta$ is the delta function, i.e., $\app{\delta}{i,j}$ is $1$ if $i=j$ and $0$ otherwise.
This study supposes that eigenvalues are sorted in ascending order like $\lambda_0\leq\cdots\leq\lambda_{N-1}$.
Note that $\lambda_0$ is strictly $0$ because all row-wise sums of $\mat{L}$ are equal to zero.
We denote a matrix spectrum $\setprn{\lambda_k}_{k=0,\ldots,N-1}$ by $\app{\sigma}{\mat{L}}$.

\newword{A graph Fourier transform (GFT)}~\citep{Taubin1995,Shuman2013,Shuman2016,Loukas2016,Hammond2011,Agaskar2013} of a graph signal $\funcdoms{f}{V}{\setR}$ is $\funcdoms{\vhat{f}}{\app{\sigma}{\mat{L}}}{\setC}$ defined by
\begin{equation}\label{eq:def_GFT}
  \app{\vhat{f}}{\lambda_k}=\iprod{f}{u_k}=\sum_{i=0}^{N-1}\app{f}{i}\wideconj{\app{u_k}{i}}
\end{equation}
for $k=0,\ldots,N-1$.
Considering the functions $\rawcurlybrace{u_k}_{k=0,\ldots,N-1}$ as signals on $G$, the GFT is a signal expansion in terms of these \newword{eigensignals}.
Then an inverse GFT is given by
\begin{equation}
  \app{f}{i}=\sum_{k=0}^{N-1}\app{\vhat{f}}{\lambda_k}\app{u_{k}}{i}.
\end{equation}
Note that the GFT on the cycle graphs is equivalent to \newword{the discrete Fourier transform (DFT)}~\citep{Taubin1995}.

When a graph Laplacian has non-distinct eigenvalues, we should remember that spectrum functions generated by the GFT are not well-defined, i.e., multi-valued.
Suppose a Laplacian matrix has two orthonormal eigensignals $u$ and $u^\prime$ corresponding to the same eigenvalue $\lambda$.
Then, for any signal $f$ on the graph, the spectral component $\app{\vhat{f}}{\lambda}$ is double-valued by $\iprod{f}{u}$ and $\iprod{f}{u^\prime}$.

In~\citep{Sandryhaila2013,Sandryhaila2014a,Sandryhaila2014}, an alternative GFT is introduced from the algebraic signal processing (ASP) approach (see~\citep{Puschel2012} for an overview of ASP).
This type of GFT expands a graph signal in terms of the eigenfunctions of an adjacency matrix.
Suppose the adjacency matrix $\mat{W}$ has eigenvalues $\rawcurlybrace{\mu_k}_{k=0,\ldots,N-1}$ and corresponding normal eigenfunctions $\rawcurlybrace{v_k}_{k=0,\ldots,N-1}$ on the vertex set.
For a signal $f$ on the graph $G$, the ASP approach defines a GFT applied to $f$ as a spectrum $\vhat{f}$ on $\app{\sigma}{\mat{W}}$ satisfying
\begin{equation}
  \app{f}{i}=\sum_{k=0}^{N-1}\app{\vhat{f}}{\mu_k}\app{v_k}{i}
\end{equation}
for $i=0,\ldots,N-1$.

In this study, we refer to this transform as \newword{the adjacency-based GFT} and refer to the transform in \cref{eq:def_GFT} as \newword{the Laplacian-based GFT} or simply, GFT.
These two GFTs usually give different spectra for the same graph signal.
We mainly discuss the Laplacian-based GFT and its 2-D extension, but occasionally we refer to the adjacency-based GFT for comparison.

\subsection{Cartesian product graph}\label{subsec:cpgraph}

\newword{A Cartesian product $G_1\cprod G_2$ of graphs $G_1=\seqprn{V_1,E_1,w_1}$ and $G_2=\seqprn{V_2,E_2,w_2}$} is a graph with vertex set $V_1\times V_2$, edge set $E$ satisfying
\begin{align}
  &\setprn{\seqprn{i_1,i_2},\seqprn{j_1,j_2}}\in E
  \\
  &\iff\sqbracket{\setprn{i_1,j_1}\in E_1,\:i_2=j_2}\OR\sqbracket{i_1=j_1,\:\setprn{i_2,j_2}\in E_2},
\end{align}
and weight function $w$ defined by
\begin{align}
  \app{w}{\seqprn{i_1,i_2},\seqprn{j_1,j_2}}
  =
  \app{w_1}{i_1,j_1}\app{\delta}{i_2,j_2}
  +\app{\delta}{i_1,j_1}\app{w_2}{i_2,j_2}.
\end{align}
The graphs $G_1$ and $G_2$ are called \newword{factor graphs of $G_1\cprod G_2$}.
\Cref{fig:Cartesian_product} shows one example of a Cartesian product operation.
For other examples, if factor graphs are \newword{cycle graphs}, their product is \newword{a two-dimensional torus graph};
if factor graphs are \newword{path graphs}, their product is \newword{a two-dimensional grid graph}.

\begin{figure}
  \centering
  \begin{subfigure}[b]{0.2\linewidth}
    \centering
    \includegraphics[width=0.7\linewidth]{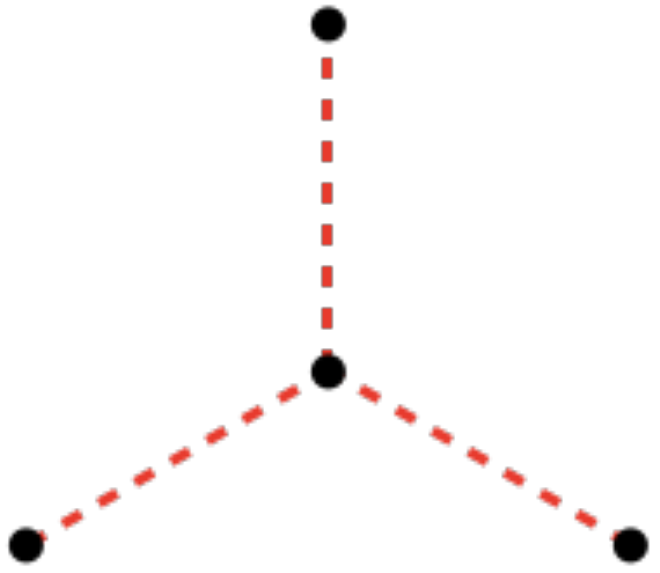}
    \caption{\vphantom{$G_1$}}
    \label{fig:factor_G1}
  \end{subfigure}
  \begin{subfigure}[b]{0.2\linewidth}
    \centering
    \includegraphics[width=0.6\linewidth]{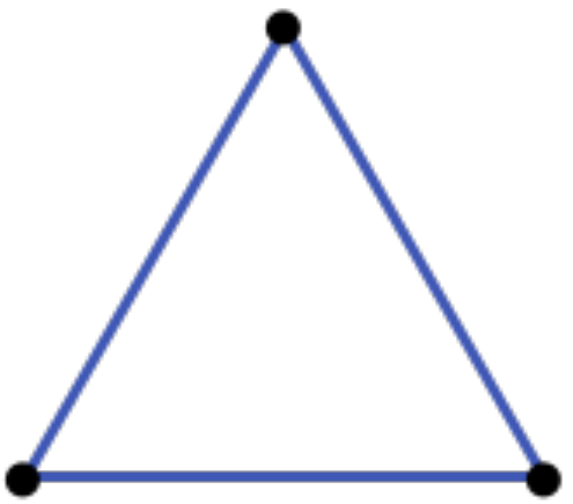}
    \caption{\vphantom{$G_2$}}
    \label{fig:factor_G2}
  \end{subfigure}
  \begin{subfigure}[t]{0.5\linewidth}
    \centering
    \includegraphics[width=0.9\linewidth]{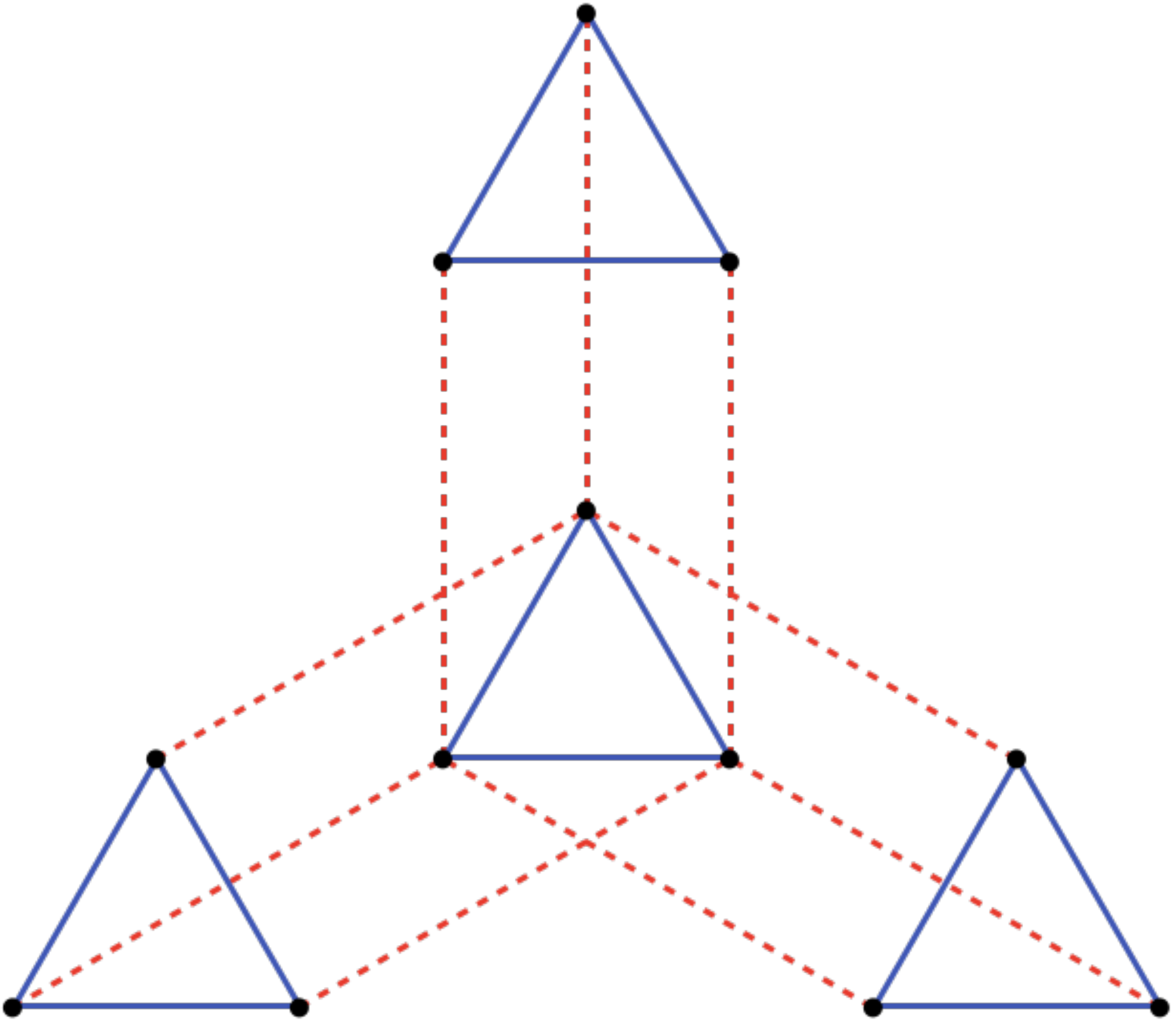}
    \caption{\vphantom{$G_1\cprod G_2$}}
    \label{fig:product_G}
  \end{subfigure}
  \caption{\subref{fig:factor_G1}\subref{fig:factor_G2} Graphs and \subref{fig:product_G} their Cartesian product graph. Red dotted edges in \subref{fig:product_G} come from \subref{fig:factor_G1} and blue solid edges come from \subref{fig:factor_G2}.}
  \label{fig:Cartesian_product}
\end{figure}

An adjacency matrix, degree matrix, and Laplacian matrix of a Cartesian product graph can be represented by those of its factor graphs.
For $n=1,2$, suppose that a factor graph $G_n$ with vertex set $V_n=\setprn{0,\ldots,N_n-1}$ has the adjacency matrix $\mat{W}_n$, the degree matrix $\mat{D}_n$, and the Laplacian matrix $\mat{L}_n$.
Then, ordering the vertices in $V_1\times V_2$ lexicographically, i.e., like $\seqprn{0,0},\seqprn{0,1},\seqprn{0,2},\ldots,\seqprn{N_1-1,N_2-1}$, the adjacency matrix, degree matrix, and Laplacian matrix of $G_1\cprod G_2$ are expressed as $\mat{W}_1\ksum\mat{W}_2$, $\mat{D}_1\ksum\mat{D}_2$, and $\mat{L}_1\ksum\mat{L}_2$, respectively.
Here, the operator $\ksum$ is \newword{a Kronecker sum} defined by $\mat{A}\ksum\mat{B}=\mat{A}\kprod\matid[n]+\matid[m]\kprod\mat{B}$ for matrices $\mat{A}\in\setMatrix{m}{m}{\setR}$ and $\mat{B}\in\setMatrix{n}{n}{\setR}$, where $\matid[n]$ is the identity matrix of size $n$.

A desirable property of the Kronecker sum allows an eigenproblem involving a Laplacian matrix of a product graph to be broken down into eigenproblems involving that of each factor graph.
Supposing the Laplacian matrix $\mat{L}_n$ has nonnegative eigenvalues $\rawcurlybrace{\uid{\lambda_k}{n}}_{k=0,\ldots,N_n-1}$ and orthonormal eigenfunctions $\rawcurlybrace{\uid{u_k}{n}}_{k=0,\ldots,N_n-1}$ for $n=1,2$, the Kronecker sum $\mat{L}_1\ksum\mat{L}_2$ has an eigenvalue $\uid{\lambda_{k_1}}{1}+\uid{\lambda_{k_2}}{2}$ and the corresponding eigenfunction $\funcdoms{\uid{u_{k_1}}{1}\uid{u_{k_2}}{2}}{V_1\times V_2}{\setC}$ satisfying
\begin{align}
  &\paren{\mat{L}_1\ksum\mat{L}_2}
  \begin{pmatrix}
    \app{\uid{u_{k_1}}{1}}{0}\app{\uid{u_{k_2}}{2}}{0}
    \\[4pt]
    \app{\uid{u_{k_1}}{1}}{0}\app{\uid{u_{k_2}}{2}}{1}
    \\
    \vdots
    \\
    \app{\uid{u_{k_1}}{1}}{N_1-1}\app{\uid{u_{k_2}}{2}}{N_2-1}
  \end{pmatrix}
  \\
  &=
  \paren{\uid{\lambda_{k_1}}{1}+\uid{\lambda_{k_2}}{2}}
  \begin{pmatrix}
    \app{\uid{u_{k_1}}{1}}{0}\app{\uid{u_{k_2}}{2}}{0}
    \\[4pt]
    \app{\uid{u_{k_1}}{1}}{0}\app{\uid{u_{k_2}}{2}}{1}
    \\
    \vdots
    \\
    \app{\uid{u_{k_1}}{1}}{N_1-1}\app{\uid{u_{k_2}}{2}}{N_2-1}
  \end{pmatrix}
\end{align}
for any $k_1=0,\ldots,N_1-1$ and $k_2=0,\ldots,N_2-1$.
The eigenvalues $\rawcurlybrace{\uid{\lambda_{k_1}}{1}+\uid{\lambda_{k_2}}{2}}_{k_1,k_2}$ are nonnegative and the eigenfunctions $\rawcurlybrace{\uid{u_{k_1}}{1}\uid{u_{k_2}}{2}}_{k_1,k_2}$ are orthonormal.
These are easily deduced from several basic properties of the Kronecker product (see e.g.~\citep[chap. 13]{Laub2005}).

An eigenproblem concerning an adjacency matrix of a product graph can be broken down in the same way.
Supposing the adjacency matrix $\mat{W}_n$ has eigenvalues $\rawcurlybrace{\uid{\mu_k}{n}}_{k=0,\ldots,N_n-1}$ and orthonormal eigenfunctions $\rawcurlybrace{\uid{v_k}{n}}_{k=0,\ldots,N_n-1}$ for $n=1,2$, the Kronecker sum $\mat{W}_1\ksum\mat{W}_2$ has eigenvalues $\rawcurlybrace{\uid{\mu_{k_1}}{1}+\uid{\mu_{k_2}}{2}}_{k_1,k_2}$ and corresponding eigenfunctions $\rawcurlybrace{\uid{v_{k_1}}{1}\uid{v_{k_2}}{2}}_{k_1,k_2}$ on $V_1\times V_2$.

\section{Multi-dimensional graph signal transforms}\label{sec:mgft}

\subsection{Multi-dimensional graph Fourier transform}\label{subsec:mgft}

Consider the Laplacian-based GFT on a Cartesian product graph.
For $n=1,2$, let $G_n$ be an undirected weighted graph with vertex set $V_n=\setprn{0,\ldots,N_n-1}$, and suppose that its graph Laplacian $\mat{L}_n$ has ascending eigenvalues $\rawcurlybrace{\uid{\lambda_k}{n}}_{k=0,\ldots,N_n-1}$ and the corresponding orthonormal eigenfunctions $\rawcurlybrace{\uid{u_k}{n}}_{k=0,\ldots,N_n-1}$.
Due to the previous discussion about product graphs, the GFT of a graph signal $\funcdoms{f}{V_1\times V_2}{\setR}$ on the product graph $G_1\cprod G_2$ is a spectrum $\funcdoms{\hat{f}}{\app{\sigma}{\mat{L}_1\ksum\mat{L}_2}}{\setC}$ given by
\begin{equation}\label{eq:GFT_on_product}
  \app{\vhat{f}}{\uid{\lambda_{k_1}}{1}+\uid{\lambda_{k_2}}{2}}
  =\sum_{i_1=0}^{N_1-1}\sum_{i_2=0}^{N_2-1}\app{f}{i_1,i_2}\wideconj{\app{\uid{u_{k_1}}{1}}{i_1}\app{\uid{u_{k_2}}{2}}{i_2}}
\end{equation}
for $k_1=0,\ldots,N_1-1$ and $k_2=0,\ldots,N_2-1$, and its inverse is
\begin{equation}\label{eq:iGFT_on_product}
  \app{f}{i_1,i_2}
  =\sum_{k_1=0}^{N_1-1}\sum_{k_2=0}^{N_2-1}\app{\vhat{f}}{\uid{\lambda_{k_1}}{1}+\uid{\lambda_{k_2}}{2}}\app{\uid{u_{k_1}}{1}}{i_1}\app{\uid{u_{k_2}}{2}}{i_2}
\end{equation}
for $i_1=0,\ldots,N_1-1$ and $i_2=0,\ldots,N_2-1$.

Considering the GFT on a product graph above, it seems natural to define the spectrum not as a univariate function on $\app{\sigma}{\mat{L}_1\ksum\mat{L}_2}$, but as a bivariate function on $\app{\sigma}{\mat{L}_1}\times\app{\sigma}{\mat{L}_2}$.
Now, we regard a signal on a Cartesian product graph as a ``two-dimensional signal'' and propose a two-dimensional GFT which gives a ``two-dimensional spectrum.''

\begin{definition}[Two-dimensional graph Fourier transform]
  A two-dimensional GFT (2-D GFT) of a signal $\funcdoms{f}{V_1\times V_2}{\setR}$ on a Cartesian product graph $G_1\cprod G_2$ is a spectrum $\funcdoms{\vhat{f}}{\app{\sigma}{\mat{L}_1}\times\app{\sigma}{\mat{L}_2}}{\setC}$ defined by
  \begin{equation}
    \app{\vhat{f}}{\uid{\lambda_{k_1}}{1},\uid{\lambda_{k_2}}{2}}
    =\sum_{i_1=0}^{N_1-1}\sum_{i_2=0}^{N_2-1}\app{f}{i_1,i_2}\wideconj{\app{\uid{u_{k_1}}{1}}{i_1}\app{\uid{u_{k_2}}{2}}{i_2}}
  \end{equation}
  for $k_1=0,\ldots,N_1-1$ and $k_2=0,\ldots,N_2-1$, and its inverse is given by
  \begin{equation}
    \app{f}{i_1,i_2}
    =\sum_{k_1=0}^{N_1-1}\sum_{k_2=0}^{N_2-1}\app{\vhat{f}}{\uid{\lambda_{k_1}}{1},\uid{\lambda_{k_2}}{2}}\app{\uid{u_{k_1}}{1}}{i_1}\app{\uid{u_{k_2}}{2}}{i_2}
  \end{equation}
  for $i_1=0,\ldots,N_1-1$ and $i_2=0,\ldots,N_2-1$.
\end{definition}

Note that the 2-D GFT is represented as a chain of matrix-matrix multiplications.
By using $N_1\times N_2$ matrices $\mat{F}=\rawparen{\app{f}{i_1,i_2}}_{i_1,i_2}$ and $\vhat{\mat{F}}=\rawparen{\app{\vhat{f}}{\vuid{\lambda_{k_1}}{1},\vuid{\lambda_{k_2}}{2}}}_{k_1,k_2}$, the 2-D GFT applied to the signal $f$ is expressed as
\begin{align}
  \vhat{\mat{F}}=\adjoint{\mat{U}_1}\mat{F}\wideconj{\mat{U}}_2,\label{eq:2DGFT_mat}
\end{align}
where $\mat{U}_n$ is an $N_n\times N_n$ unitary matrix with $\seqprn{i,k}$-th element $\app{\uid{u_k}{n}}{i}$ for $n=1,2$.
Then its inverse is given by
\begin{align}
  \mat{F}=\mat{U}_1\vhat{\mat{F}}\trsps{\mat{U}_2}.\label{eq:i2DGFT_mat}
\end{align}

The 2-D GFT is related to existing transformations as follows.
First, when both factor graphs are cycle graphs, the 2-D GFT can be equivalent to the 2-D DFT.
Second, when one factor graph is a cycle graph expressing a periodic time axis, the 2-D GFT is called \newword{a joint graph and temporal Fourier transform}~\citep{Loukas2016}.
The proposed 2-D GFT generalizes these existing transformations.

For a signal on product graphs, the proposed 2-D GFT provides the following advantages over the conventional GFT: directional frequency analysis, multi-valuedness resolution, and reduced computational time.
The remainder of this subsection explains these advantages.

First, the 2-D GFT enables us to analyze graph signals in terms of the frequency characteristics along each factor graph.
Because the 2-D GFT parallels the 2-D Fourier transform, we expect that the $n$-th variable of 2-D spectra behaves as a ``frequency along the $n$-th factor graph.''
That is, if a spectrum $\app{\vhat{f}}{\vuid{\lambda}{1},\placeholder}$ at a large $\uid{\lambda}{1}$ is dominant, a signal $f$ should drastically fluctuate along the graph $G_1$, whereas if $\app{\vhat{f}}{\vuid{\lambda}{1},\placeholder}$ at a small $\uid{\lambda}{1}$ is dominant, $f$ should gradually change along $G_1$.

See \cref{fig:various_signals} for an example where the expectation is likely real.
It shows several signals on a product of a path graph $G_1$ and a wheel graph $G_2$ shown in \cref{fig:frequency_cpgraph} (vertex domain representation) and their power spectra obtained by our 2-D GFT (frequency domain representation).
The vertices of $G_1$ are indexed as \cref{fig:path}, and those of $G_2$ are indexed as \cref{fig:wheel}.
All edges of $G_1$, $G_2$, and $G_1\cprod G_2$ are weighted by one.
In \cref{fig:vertex_ll,fig:freq_ll}, the signal gradually changes along both $G_1$ and $G_2$ in the vertex domain, and the spectrum generated by its 2-D GFT indicates that.
In \cref{fig:vertex_lh}, the signal gradually changes along $G_1$; however, for many $i_1\in V_1$, the signal at the center vertex $\seqprn{i_1,0}$ greatly differs from the signal at the surrounding vertices $\setprnsep{\seqprn{i_1,i_2}}{i_2\neq 0}$.
Therefore, the signal has low-frequency along $G_1$ and high-frequency along $G_2$.
The spectrum shown in \cref{fig:freq_lh} explains these anisotropic signal characteristics.
The spectra in \cref{fig:freq_hl,fig:freq_hh} also indicate the directional characteristics of the signals in \cref{fig:vertex_hl,fig:vertex_hh}, respectively.
Therefore, in these cases we can see $\uid{\lambda}{1}\in\app{\sigma}{\mat{L}_1}$ as a frequency along $G_1$ and $\uid{\lambda}{2}\in\app{\sigma}{\mat{L}_2}$ as a frequency along $G_2$.

\begin{figure*}
  \centering
  \begin{subfigure}[b]{0.245\linewidth}
    \centering
    \includegraphics[width=0.9\linewidth]{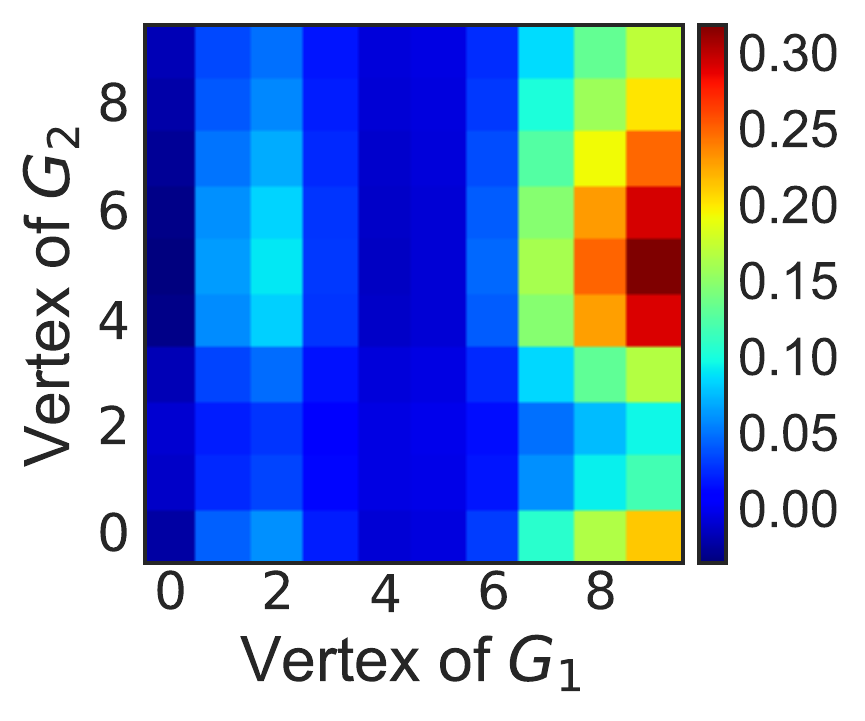}
    \hspace{-10pt}
    \caption{\vphantom{Vertex domain}}
    \label{fig:vertex_ll}
  \end{subfigure}
  \begin{subfigure}[b]{0.245\linewidth}
    \centering
    \includegraphics[width=0.9\linewidth]{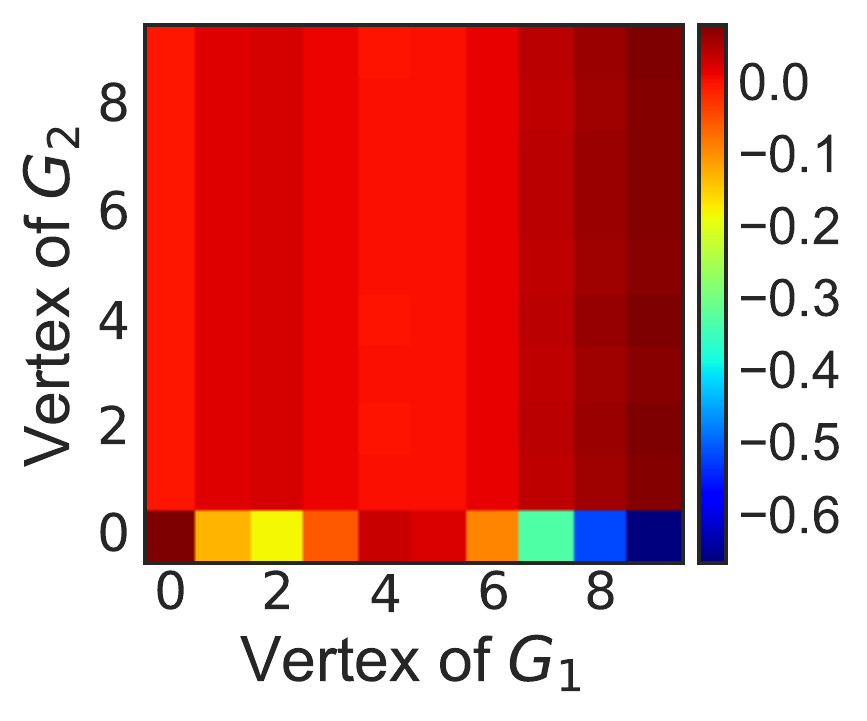}
    \hspace{-10pt}
    \caption{\vphantom{Vertex domain}}
    \label{fig:vertex_lh}
  \end{subfigure}
  \begin{subfigure}[b]{0.245\linewidth}
    \centering
    \includegraphics[width=0.9\linewidth]{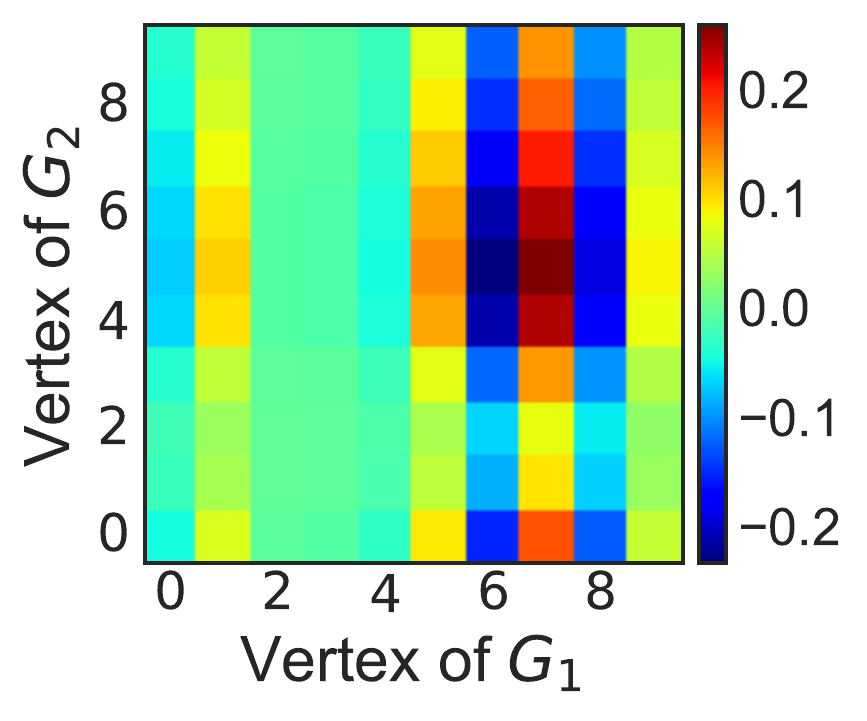}
    \hspace{-10pt}
    \caption{\vphantom{Vertex domain}}
    \label{fig:vertex_hl}
  \end{subfigure}
  \begin{subfigure}[b]{0.245\linewidth}
    \centering
    \includegraphics[width=0.9\linewidth]{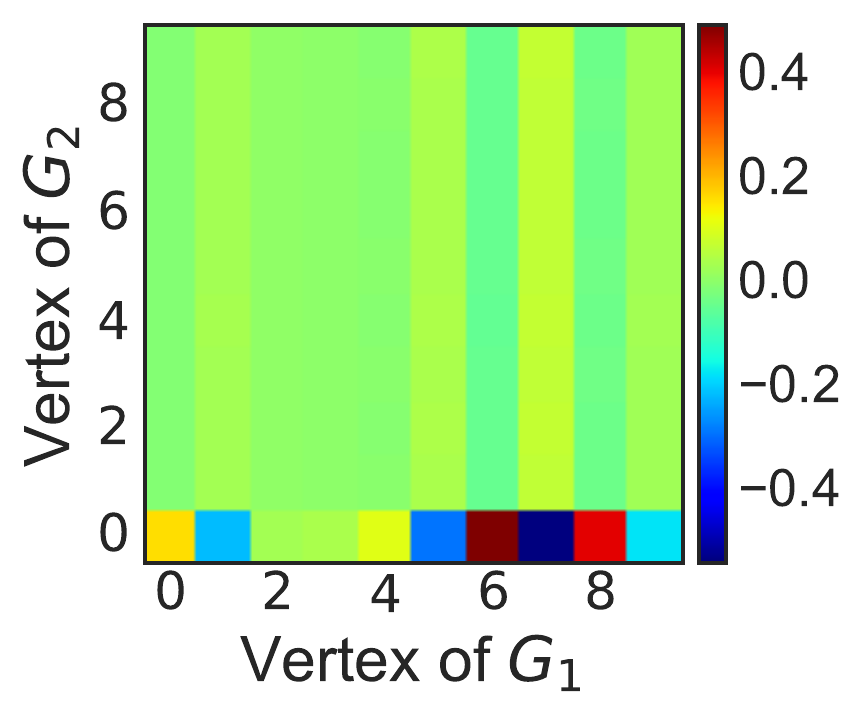}
    \hspace{-10pt}
    \caption{\vphantom{Vertex domain}}
    \label{fig:vertex_hh}
  \end{subfigure}
  \\[8pt]
  \begin{subfigure}[b]{0.245\linewidth}
    \centering
    \includegraphics[width=\linewidth]{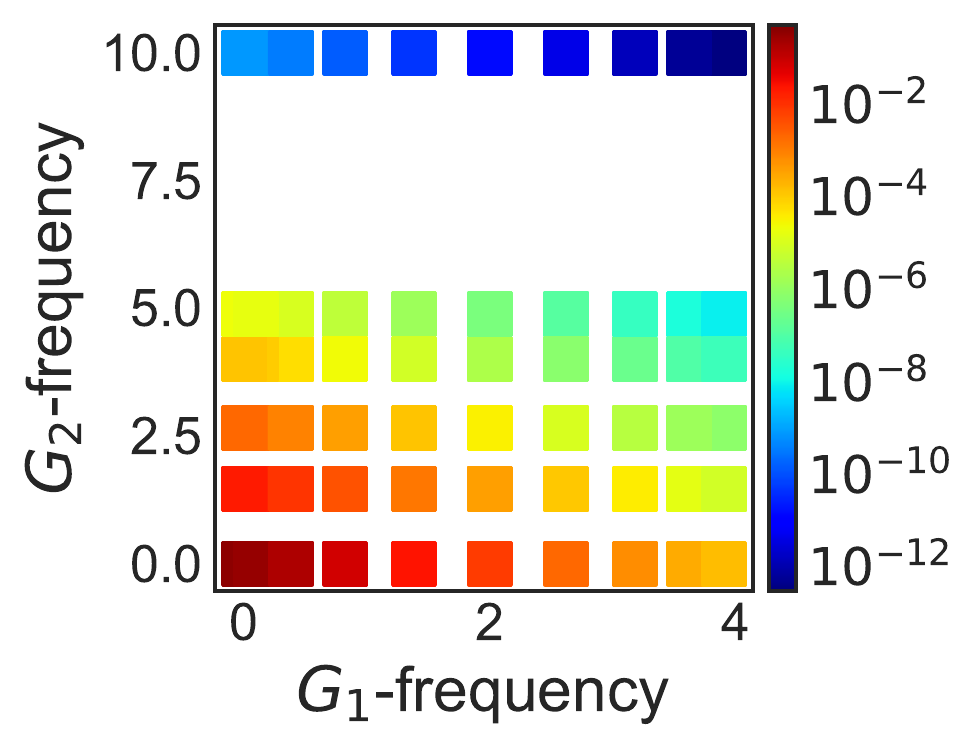}
    \hspace{-10pt}
    \caption{\vphantom{Frequency domain}}
    \label{fig:freq_ll}
  \end{subfigure}
  \begin{subfigure}[b]{0.245\linewidth}
    \centering
    \includegraphics[width=\linewidth]{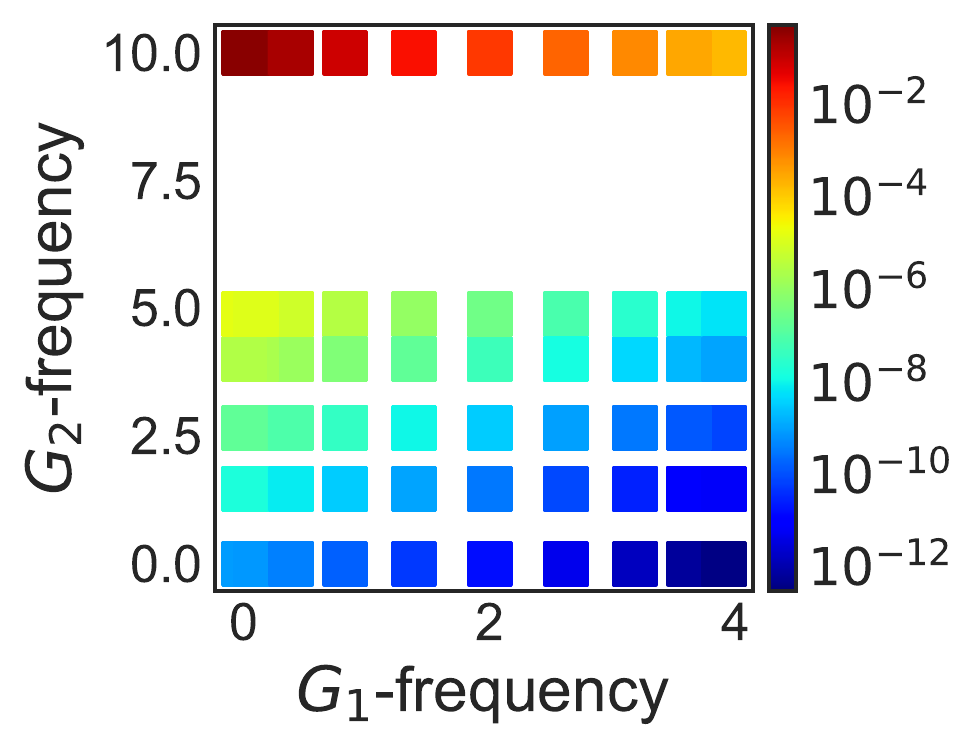}
    \hspace{-10pt}
    \caption{\vphantom{Frequency domain}}
    \label{fig:freq_lh}
  \end{subfigure}
  \begin{subfigure}[b]{0.245\linewidth}
    \centering
    \includegraphics[width=\linewidth]{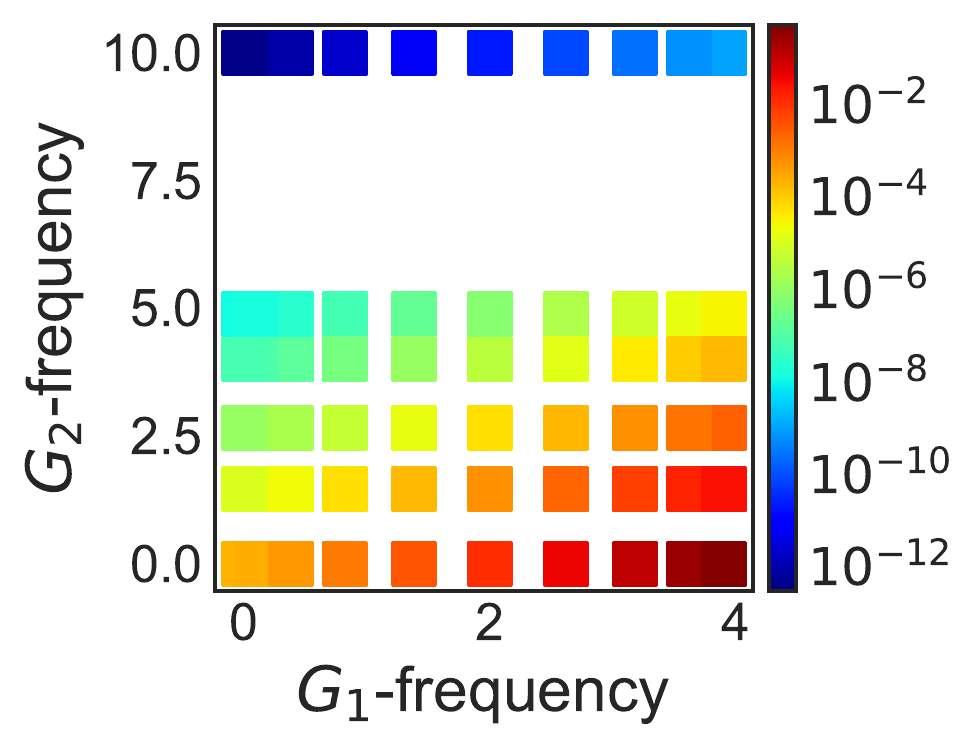}
    \hspace{-10pt}
    \caption{\vphantom{Frequency domain}}
    \label{fig:freq_hl}
  \end{subfigure}
  \begin{subfigure}[b]{0.245\linewidth}
    \centering
    \includegraphics[width=\linewidth]{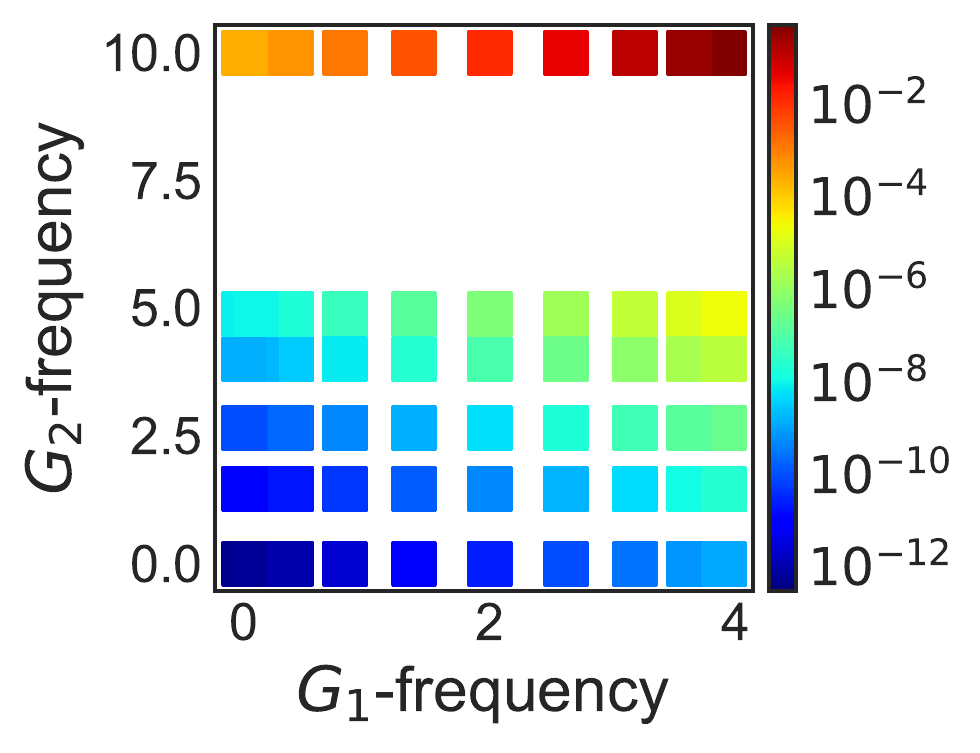}
    \hspace{-10pt}
    \caption{\vphantom{Frequency domain}}
    \label{fig:freq_hh}
  \end{subfigure}
  \caption{(a--d) Various signals on the product graph $G_1\cprod G_2$ shown in \cref{fig:pathxwheel}, and (e--h) their power spectra obtained by the 2-D GFT. \subref{fig:vertex_ll} illustrates low-$G_1$-frequency and low-$G_2$-frequency signal, \subref{fig:vertex_lh} illustrates low-$G_1$-frequency and high-$G_2$-frequency signal, \subref{fig:vertex_hl} illustrates high-$G_1$-frequency and low-$G_2$-frequency signal, and \subref{fig:vertex_hh} illustrates high-$G_1$-frequency and high-$G_2$-frequency signal. In each signal, the intensity at $\seqprn{i_1,i_2}\in V_1\times V_2$ is indicated by the color of the element $\seqprn{i_1,i_2}$. (e--h) illustrates the power spectra of (a--d) obtained by the proposed MGFT, respectively. In each power spectrum, the intensity at $\seqprn{\uidfootnote{\lambda}{1},\uidfootnote{\lambda}{2}}\in\app{\sigma}{\mat{L}_1}\times\app{\sigma}{\mat{L}_2}$ is indicated by color of the square at $\seqprn{\uidfootnote{\lambda}{1},\uidfootnote{\lambda}{2}}$.}
  \label{fig:various_signals}
\end{figure*}

\begin{figure}
  \centering
  \begin{subfigure}[b]{0.12\linewidth}
    \centering
    \includegraphics[width=\linewidth]{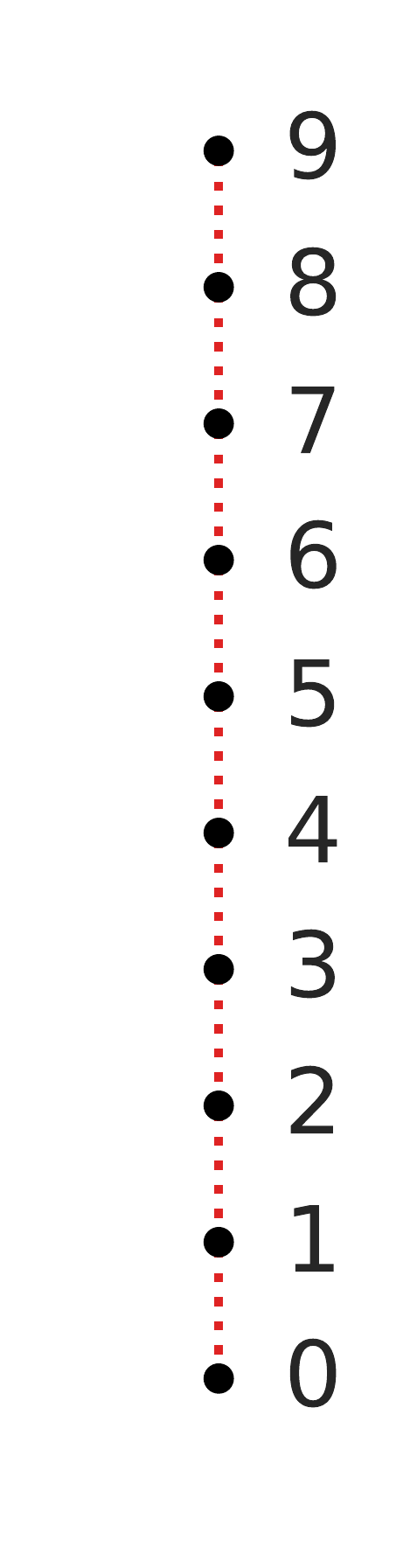}
    \caption{\vphantom{$G_1$}}
    \label{fig:path}
  \end{subfigure}
  \begin{subfigure}[b]{0.42\linewidth}
    \centering
    \includegraphics[width=\linewidth]{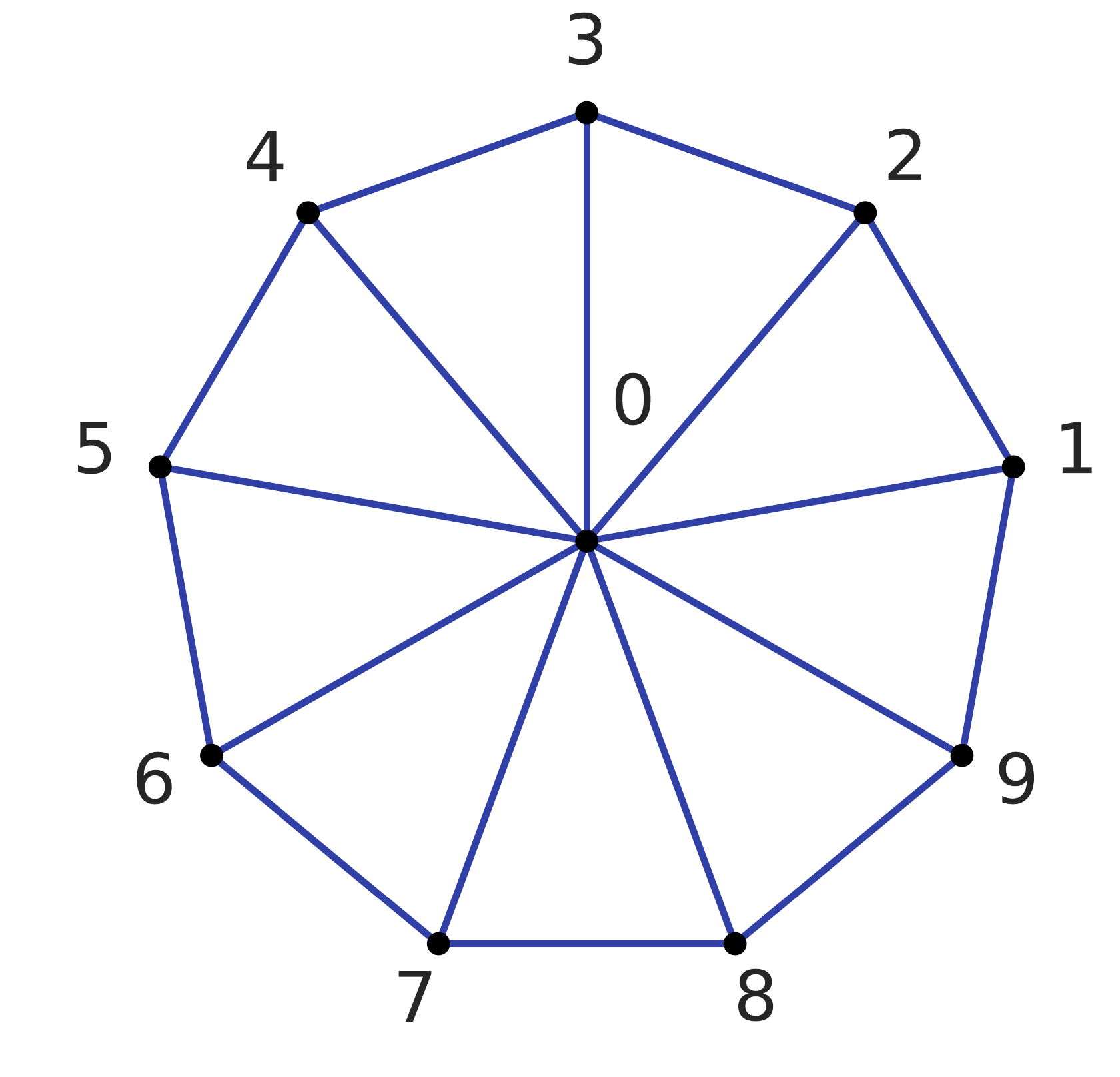}
    \caption{\vphantom{$G_2$}}
    \label{fig:wheel}
  \end{subfigure}
  \begin{subfigure}[t]{0.42\linewidth}
    \centering
    \includegraphics[width=\linewidth]{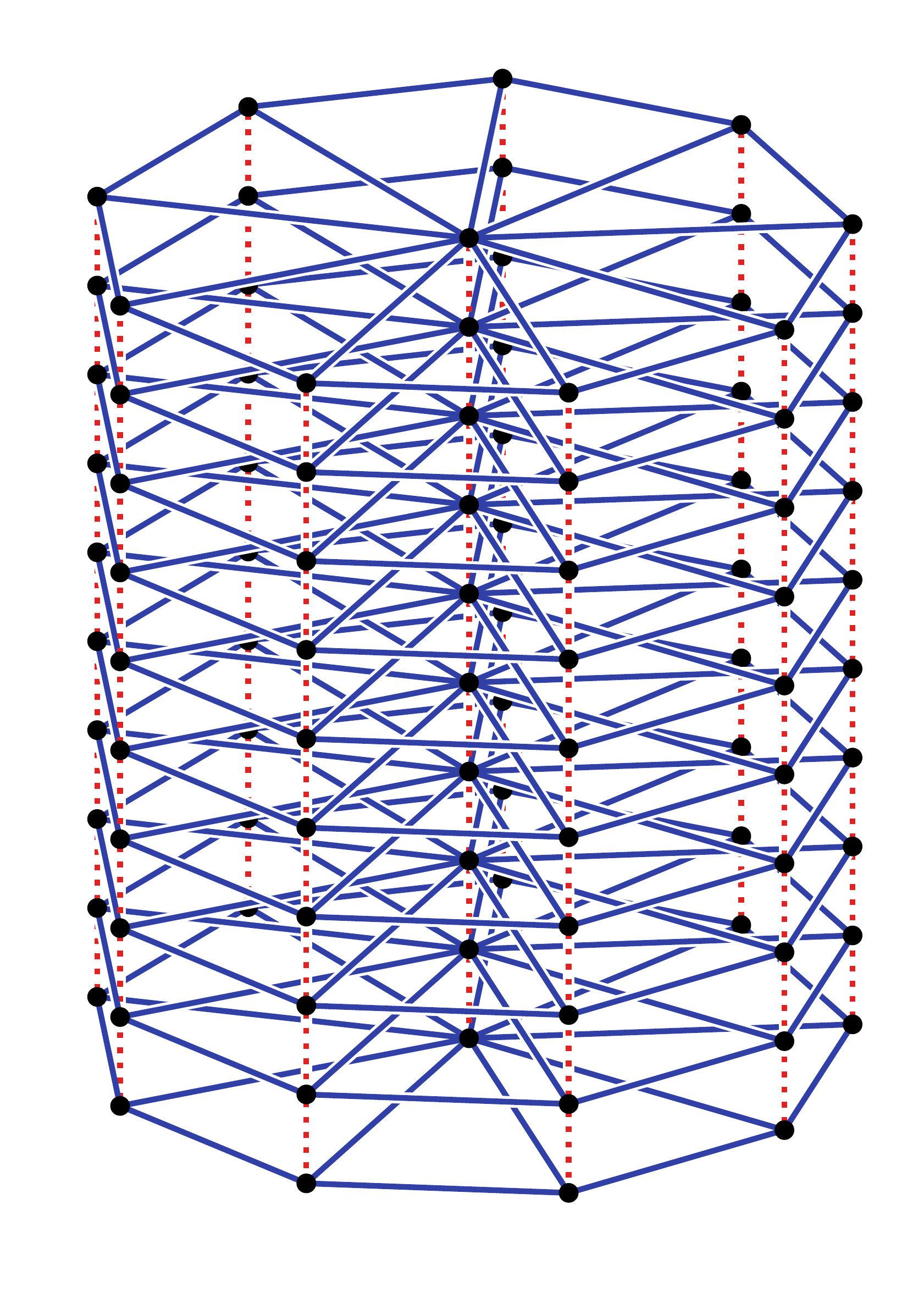}
    \caption{\vphantom{$G_1\cprod G_2$}}
    \label{fig:pathxwheel}
  \end{subfigure}
  \caption{\subref{fig:path} Path graph, \subref{fig:wheel} wheel graph, and \subref{fig:pathxwheel} their Cartesian product graph. Red dotted edges in \subref{fig:pathxwheel} come from \subref{fig:path} and blue solid edges come from \subref{fig:wheel}.}
  \label{fig:frequency_cpgraph}
\end{figure}

Theoretically, the total directional variation study below indicates that variables of the 2-D GFT spectra behave as directional frequencies along factor graphs.
The discussion here is based on the total variation study in~\citep{Shuman2013}.
For a graph signal $\funcdoms{f}{V}{\setR}$, the graph gradient of $f$ at $i\in V$ is the signal $\funcdoms{\nabla_i f}{V}{\setR}$ defined by
\begin{align}
  \app{\nabla_i f}{j}=\sqrt{\app{w}{i,j}}\paren{\app{f}{j}-\app{f}{i}}.
\end{align}
For a 2-D graph signal $\funcdoms{f}{V_1\times V_2}{\setR}$, the local $G_1$-directional variation of $f$ at $\seqprn{i_1,i_2}\in V_1\times V_2$ is a Euclidean norm of $G_1$-directional components of the graph gradient $\nabla_{i_1,i_2}f$ given by
\begin{align}
  \app{\uid{\mathcal{V}_{i_1,i_2}}{1}}{f}
  &=\sqrt{\sum_{j_1=0}^{N_1-1}\absprn{\app{\nabla_{i_1,i_2}f}{j_1,i_2}}^2}
  \\
  &=\sqrt{\sum_{j_1=0}^{N_1-1}\app{w_1}{i_1,j_1}\paren{\app{f}{j_1,i_2}-\app{f}{i_1,i_2}}^2},
\end{align}
which indicates how intensely signal $f$ changes along $G_1$ at the vertex $\seqprn{i_1,i_2}$;
and the total $G_1$-directional variation of $f$ is a 2-Dirichlet form of the local $G_1$-directional variations given by
\begin{align}
  \app{\uid{S_2}{1}}{f}
  &=\frac{1}{2}\sum_{i_1=0}^{N_1-1}\sum_{i_2=0}^{N_2-1}\absprn{\app{\uid{\mathcal{V}_{i_1,i_2}}{1}}{f}}^2
  \\
  &=\frac{1}{2}\sum_{i_1=0}^{N_1-1}\sum_{j_1=0}^{N_1-1}\app{w_1}{i_1,j_1}\norm[2]{\app{f}{j_1,\placeholder}-\app{f}{i_1,\placeholder}}^2,
\end{align}
which indicates how intensely the signal $f$ changes along $G_1$.
Using the matrix representation of the 2-D GFT in \cref{eq:2DGFT_mat}, the total $G_1$-directional variation can be decomposed as
\begin{align}
  \app{\uid{S_2}{1}}{f}
  &=\sum_{i_1=0}^{N_1-1}\sum_{j_1=0}^{N_1-1}\app{w_1}{i_1,j_1}\norm[2]{\app{f}{i_1,\placeholder}}^2
  \\
  &\hphantom{=}-\sum_{i_1=0}^{N_1-1}\sum_{j_1=0}^{N_1-1}\app{w_1}{i_1,j_1}\iprod{\app{f}{i_1,\placeholder}}{\app{f}{j_1,\placeholder}}
  \\
  &=\app{\tr}{\vtrsps{\mat{F}}\mat{D}_1\mat{F}}
  -\app{\tr}{\vtrsps{\mat{F}}\mat{W}_1\mat{F}}
  \\
  &=\app{\tr}{\vtrsps{\mat{F}}\mat{L}_1\mat{F}}
  =\sum_{k_1=0}^{N_1-1}\uid{\lambda_{k_1}}{1}\norm[2]{\app{\vhat{f}}{\vuid{\lambda_{k_1}}{1},\placeholder}}^2,
\end{align}
where $\iprod{\app{f}{i_1,\placeholder}}{\app{f}{j_1,\placeholder}}$ is $\sum_{i_2=0}^{N_2-1}\app{f}{i_1,i_2}\app{f}{j_1,i_2}$.
This decomposition shows that the higher-$G_1$-frequency components of signals contribute more to their variation along $G_1$.
The total quadratic $G_2$-directional variation $\app{\uid{S_2}{2}}{f}$ is decomposed as $\app{\uid{S_2}{2}}{f}=\sum_{k_2=0}^{N_2-1}\uid{\lambda_{k_2}}{2}\norm[2]{\app{\vhat{f}}{\placeholder,\vuid{\lambda_{k_2}}{2}}}^2$ in the same way, thus the higher-$G_2$-frequency components of signals contribute more to their variation along $G_2$.

A total directional variation also appears in other situations.
Laplacian eigenmaps, a popular tool used in manifold learning, find the low-dimensional data representation that minimizes the total variation along a similarity graph of data~\citep{Belkin2001}.
The minimizer provides the smoothest representation on that graph.
Matrix completion on graphs infers the original matrix from its incomplete observation when given a column-wise-similarity graph $G_C$ and a row-wise-similarity graph $G_R$~\citep{Kalofolias2014}.
The method minimizes the sum of four terms: the distance between the inference and the observation, nuclear norm of the inference, total $G_C$-directional variation, and total $G_R$-directional variation.
Its minimizer is close to the observation, is low-rank, and is smooth along $G_C$ and $G_R$.

The second advantage of the 2-D GFT is that it can solve the multi-valuedness of the ordinary GFT.
When a graph Laplacian has non-distinct eigenvalues, the ordinary GFT spectra of graph signals are multi-valued at multiple eigenvalues.
Therefore, for a signal on a product graph $G_1\cprod G_2$, if frequencies $\uid{\lambda_{k_1}}{1}+\uid{\lambda_{k_2}}{2}$ and $\uid{\lambda_{l_1}}{1}+\uid{\lambda_{l_2}}{2}$ are equal with $k_1\neq l_1$ or $k_2\neq l_2$, the GFT assigns two different values to the signal spectrum at the frequency.
However, even if $\uid{\lambda_{k_1}}{1}+\uid{\lambda_{k_2}}{2}$ and $\uid{\lambda_{l_1}}{1}+\uid{\lambda_{l_2}}{2}$ are equal, pairs of frequencies $\seqprn{\vuid{\lambda_{k_1}}{1},\vuid{\lambda_{k_2}}{2}}$ and $\seqprn{\vuid{\lambda_{l_1}}{1},\vuid{\lambda_{l_2}}{2}}$ may be different, and then the 2-D spectrum by the 2-D GFT is well-defined at the frequency pairs.
A product graph $G\cprod G$ whose factor graph $G$ has distinct eigenvalues serves as a typical example.
For any $k\neq l$, its 1-D spectrum is double-valued at frequency $\lambda_k+\lambda_l=\lambda_l+\lambda_k$ but 2-D spectrum is well-defined at $\seqprn{\lambda_k,\lambda_l}$ and $\seqprn{\lambda_l,\lambda_k}$ separately.

The third advantage is that the 2-D GFT takes less computational time than the ordinary GFT.
The 2-D GFT and its inverse on a product graph $G_1\cprod G_2$ cost $\app{O}{N_1^2N_2+N_1N_2^2}$ time with the straightforward matrix multiplication, although the conventional GFT and its inverse on the same graph cost $\app{O}{N_1^2N_2^2}$ time.
Preliminarily, the 2-D GFT and its inverse need an eigendecomposition of the graph Laplacians of factor graphs $G_1$ and $G_2$ that costs $\app{O}{N_1^3+N_2^3}$ time, whereas the conventional GFT and its inverse need an eigendecomposition of a graph Laplacian of $G$ that costs $\app{O}{N_1^3N_2^3}$ time.
Furthermore, we can theoretically reduce the temporal cost of the 2-D transform by utilizing fast matrix multiplication algorithms, such as the Strassen algorithm, because these operations consist of a chain of matrix-matrix multiplications.

For an arbitrary natural number $n$, an $n$-dimensional GFT on a Cartesian product graph
\begin{equation}
  G_1\cprod\cdots\cprod G_n=\paren{\paren{\cdots\paren{G_1\cprod G_2}\cdots}\cprod G_n}
\end{equation}
is inductively defined by 2-D GFTs.

\subsection{Adjacency-based multi-dimensional graph Fourier transform}\label{subsec:asp_mgft}
We can extend the adjacency-based GFT to its 2-D version in the same manner as the Laplacian-based GFT, because the adjacency matrix of a Cartesian product graph is also a Kronecker sum of those of factor graphs, such as the Laplacian matrix.

Consider the adjacency-based GFT on a Cartesian product graph first.
For $n=1,2$, let $G_n$ be a weighted graph with vertex set $V=\setprn{0,\ldots,N_n-1}$, and suppose that its adjacency matrix $\mat{W}_n$ has eigenvalues $\rawcurlybrace{\uid{\mu_k}{n}}_{k=0,\ldots,N_n-1}$ and the corresponding eigenfunctions $\rawcurlybrace{\uid{v_k}{n}}_{k=0,\ldots,N_n-1}$ on $V_n$.
From the discussion about product graphs in \cref{subsec:cpgraph}, an adjacency-based GFT of a signal $f$ on $G_1\cprod G_2$ is a spectrum $\vhat{f}$ on $\app{\sigma}{\mat{W}_1\ksum\mat{W}_2}$ satisfying
\begin{equation}\label{eq:adj_GFT_on_cp_graph}
  \app{f}{i_1,i_2}
  =\sum_{k_1=0}^{N_1-1}\sum_{k_2=0}^{N_2-1}\app{\vhat{f}}{\uid{\mu_{k_1}}{1}+\uid{\mu_{k_2}}{2}}\app{\uid{v_{k_1}}{1}}{i_1}\app{\uid{v_{k_2}}{2}}{i_2}
\end{equation}
for $i_1=0,\ldots,N_1-1$ and $i_2=0,\ldots,N_2-1$.
Based on that, we define an adjacency-based 2-D GFT of a signal $f$ on $G_1\cprod G_2$ is a spectrum $\vhat{f}$ on $\app{\sigma}{\mat{W}_1}\times\app{\sigma}{\mat{W}_2}$ satisfying
\begin{equation}
  \app{f}{i_1,i_2}
  =\sum_{k_1=0}^{N_1-1}\sum_{k_2=0}^{N_2-1}\app{\vhat{f}}{\uid{\mu_{k_1}}{1},\uid{\mu_{k_2}}{2}}\app{\uid{v_{k_1}}{1}}{i_1}\app{\uid{v_{k_2}}{2}}{i_2}
\end{equation}
for $i_1=0,\ldots,N_1-1$ and $i_2=0,\ldots,N_2-1$, and then call an eigenvalue of $\mat{W}_1$ a frequency along $G_1$-direction and that of $\mat{W}_2$ a frequency along $G_2$-direction.

As with the Laplacian-based 2-D GFT, the adjacency-based 2-D GFT and its inverse cost $\app{O}{N_1^2N_2+N_1N_2^2}$ time in transforms itself and $\app{O}{N_1^3+N_2^3}$ time in the preliminary eigendecomposition.

\section{Multi-dimensional graph signal filtering}\label{sec:filtering}

\subsection{Graph spectral filtering}\label{subsec:spectral_filtering}
In time signal processing, \newword{a spectral filtering} is a multiplication in frequency domain.
Filtering a temporal signal with a spectrum $\vhat{f}_\IN$ by a filter with frequency response $\vhat{h}$, we obtain a signal with a spectrum \begin{equation}
  \app{\vhat{f}_\OUT}{\omega}=\app{\vhat{h}}{\omega}\app{\vhat{f}_\IN}{\omega}.
\end{equation}

This filtering framework is easily extended to GSP.
Let $G$ be an undirected weighted graph with vertex set $V=\setprn{0,\ldots,N-1}$ whose graph Laplacian has eigenvalues $\rawcurlybrace{\lambda_k}_{k=0,\ldots,N-1}$.
A filter with \newword{spectral kernel} $\funcdoms{\vhat{h}}{\setRnonneg}{\setC}$, applied to a signal $\funcdoms{f_\IN}{V}{\setR}$ on the graph $G$, gives an output signal $\funcdoms{f_\OUT}{V}{\setR}$ on the graph defined by
\begin{equation}
  \app{\vhat{f}_\OUT}{\lambda_k}=\app{\vhat{h}}{\lambda_k}\app{\vhat{f}_\IN}{\lambda_k}\label{eq:spectral_filtering}
\end{equation}
where $\vhat{f}_\IN$ and $\vhat{f}_\OUT$ are spectra of $f_\IN$ and $f_\OUT$ obtained by the GFT, respectively (see~\citep{Shuman2013}).
Many graph spectral filter designs are considered: a polynomial kernel filter~\citep{Shuman2013,Shuman2016}, a heat kernel filter~\citep{Shuman2013,Zhang2008,Shuman2016}, and a graph bilateral filter~\citep{Gadde2013}.

We propose \newword{a 2-D graph spectral filter} that multiplies the 2-D GFT spectra by a 2-D spectral kernel, whereas an existing graph spectral filter multiplies the conventional GFT spectra by a 1-D spectral kernel.
For $n=1,2$, let $G_n$ be an undirected weighted graph with vertex set $V_n=\setprn{0,\ldots,N_n-1}$ whose graph Laplacian matrix $\mat{L}_n$ has ascending eigenvalues $\rawcurlybrace{\uid{\lambda_k}{n}}_{k=0,\ldots,N_n-1}$ and the corresponding orthonormal eigenfunctions $\rawcurlybrace{\uid{u_k}{n}}_{k=0,\ldots,N_n-1}$ on $V_n$.

\begin{definition}[Two-dimensional graph spectral filtering]
  A two-dimensional graph spectral filtering of a graph signal $\funcdoms{f_\IN}{V_1\times V_2}{\setR}$ on a Cartesian product graph $G_1\cprod G_2$ with spectral kernel $\funcdoms{\vhat{h}}{\setRnonneg\times\setRnonneg}{\setC}$ gives a graph signal $\funcdoms{f_\OUT}{V_1\times V_2}{\setC}$ with 2-D spectrum
  \begin{equation}
    \app{\vhat{f}_\OUT}{\uid{\lambda_{k_1}}{1},\uid{\lambda_{k_2}}{2}}
    =\app{\vhat{h}}{\uid{\lambda_{k_1}}{1},\uid{\lambda_{k_2}}{2}}\app{\vhat{f}_\IN}{\uid{\lambda_{k_1}}{1},\uid{\lambda_{k_2}}{2}}\label{eq:2Dspec_filtering}
  \end{equation}
  where $\vhat{f}_\IN$ is a 2-D spectrum of $f_\IN$ obtained by the 2-D GFT.
\end{definition}

In our 2-D graph spectral filtering framework, we can design directional frequency responses (unlike in 1-D graph spectral filtering).
When applied to the signal $f_\IN$, a 1-D graph spectral filter gives a signal $f_\OUT$ that has a spectrum
\begin{equation}
  \app{\vhat{f}_\OUT}{\uid{\lambda_{k_1}}{1},\uid{\lambda_{k_2}}{2}}
  =\app{\vhat{h}}{\uid{\lambda_{k_1}}{1}+\uid{\lambda_{k_2}}{2}}\app{\vhat{f}_\IN}{\uid{\lambda_{k_1}}{1},\uid{\lambda_{k_2}}{2}}
\end{equation}
where $\funcdoms{\vhat{h}}{\setRnonneg}{\setC}$ is the 1-D spectral kernel.
When applied to the same signal, two factor-graph-wise 1-D graph spectral filters give a signal $f_\OUT$ that has a spectrum
\begin{equation}
  \app{\vhat{f}_\OUT}{\uid{\lambda_{k_1}}{1},\uid{\lambda_{k_2}}{2}}
  =\app{\vhat{h}_1}{\uid{\lambda_{k_1}}{1}}\app{\vhat{h}_2}{\uid{\lambda_{k_2}}{2}}\app{\vhat{f}_\IN}{\uid{\lambda_{k_1}}{1},\uid{\lambda_{k_2}}{2}}
\end{equation}
where $\funcdoms{\vhat{h}_1}{\setRnonneg}{\setC}$ and $\funcdoms{\vhat{h}_2}{\setRnonneg}{\setC}$ are the 1-D spectral kernels.
These two filtering frameworks are less expressive than a 2-D graph spectral filtering framework.

If a 2-D graph spectral filter has a polynomial kernel, it has certain locality in the vertex domain.
Note that in 1-D graph polynomial filtering, \citet{Hammond2011} pointed out the following locality: with an $S$-degree polynomial filter, the output signal at some vertex is a linear combination of the input signal in its $S$-hop neighborhood, i.e., the set of vertices reachable through no more than $S$ edges.
Let $\funcdoms{\vhat{h}_{S_1S_2}}{\setRnonneg\times\setRnonneg}{\setC}$ be a 2-D polynomial kernel given by
\begin{equation}\label{eq:2Dpolykernel}
  \app{\vhat{h}_{S_1S_2}}{\uid{\lambda}{1},\uid{\lambda}{2}}
  =\sum_{s_1=0}^{S_1}\sum_{s_2=0}^{S_2}h_{s_1s_2}\paren{\uid{\lambda}{1}}^{s_1}\paren{\uid{\lambda}{2}}^{s_2}
\end{equation}
with $h_{00},\ldots,h_{S_1S_2}\in\setC$.
By using $N_1\times N_2$ matrices $\vhat{\mat{F}}_\IN=\rawparen{\app{\vhat{f}_\IN}{\vuid{\lambda_{k_1}}{1},\vuid{\lambda_{k_2}}{2}}}_{k_1,k_2}$ and $\vhat{\mat{F}}_\OUT=\rawparen{\app{\vhat{f}_\OUT}{\vuid{\lambda_{k_1}}{1},\vuid{\lambda_{k_2}}{2}}}_{k_1,k_2}$, the 2-D spectral filtering with the kernel $\vhat{h}_{S_1S_2}$ is represented as
\begin{equation}\label{eq:2Dfilter_freq}
  \vhat{\mat{F}}_\OUT
  =\sum_{s_1=0}^{S_1}\sum_{s_2=0}^{S_2}h_{s_1s_2}\mat{\Lambda}_1^{s_1}\vhat{\mat{F}}_\IN\mat{\Lambda}_2^{s_2}
\end{equation}
in frequency domain, where $\mat{\Lambda}_n$ is a diagonal matrix with diagonal elements $\uid{\lambda_0}{n},\ldots,\uid{\lambda_{N_n-1}}{n}$ for $n=1,2$.

\begin{proposition}[Slight generalization of {\citep[Equation~(12)]{Loukas2016}}]\label{prop:2Dfilter_vertex}
  By using $N_1\times N_2$ matrices $\mat{F}_\IN=\rawparen{\app{f_\IN}{i_1,i_2}}_{i_1,i_2}$ and $\mat{F}_\OUT=\rawparen{\app{f_\OUT}{i_1,i_2}}_{i_1,i_2}$, the filtering of \cref{eq:2Dfilter_freq} can be represented as
  \begin{equation}
    \mat{F}_\OUT
    =\sum_{s_1=0}^{S_1}\sum_{s_2=0}^{S_2}h_{s_1s_2}\mat{L}_1^{s_1}\mat{F}_\IN\mat{L}_2^{s_2}
  \end{equation}
  in vertex domain.
\end{proposition}

\begin{proof}
  Due to the matrix representation of the inverse 2-D GFT in \cref{eq:i2DGFT_mat}, we have
  \begin{align}
    \mat{F}_\OUT
    &=\mat{U}_1\paren{\sum_{s_1=0}^{S_1}\sum_{s_2=0}^{S_2}h_{s_1s_2}\mat{\Lambda}_1^{s_1}\vhat{\mat{F}}_\IN\mat{\Lambda}_2^{s_2}}\vtrsps{\mat{U}_2}
    \\
    &=\sum_{s_1=0}^{S_1}\sum_{s_2=0}^{S_2}h_{s_1s_2}\paren{\mat{U}_1\mat{\Lambda}_1^{s_1}\vadjoint{\mat{U}_1}}\paren{\mat{U}_1\vhat{\mat{F}}_\IN\vtrsps{\mat{U}_2}}\paren{\wideconj{\mat{U}}_2\mat{\Lambda}_2^{s_2}\vtrsps{\mat{U}_2}}
    \\
    &=\sum_{s_1=0}^{S_1}\sum_{s_2=0}^{S_2}h_{s_1s_2}\mat{L}_1^{s_1}\mat{F}_\IN\mat{L}_2^{s_2}.
  \end{align}
\end{proof}

The locality of a 2-D polynomial kernel filter is easily deduced from \cref{prop:2Dfilter_vertex}.
Given the 2-D polynomial kernel in \cref{eq:2Dpolykernel}, define a neighborhood $\app{\mathcal{N}}{i_1,i_2}$ of $\seqprn{i_1,i_2}\in V_1\times V_2$ as follows: a vertex $\seqprn{j_1,j_2}\in V_1\times V_2$ belongs to $\app{\mathcal{N}}{i_1,i_2}$ if and only if it is reachable from $\seqprn{i_1,i_2}$ through $t_1$ edges along $G_1$ and $t_2$ edges along $G_2$ and some nonzero $h_{s_1s_2}$ exists satisfying $t_1\leq s_1$ and $t_2\leq s_2$.
A 2-D graph signal filter with the kernel $\vhat{h}_{S_1S_2}$ is local with respect to this neighborhood.

\begin{corollary}\label{cor:2Dfilter_locality}
  In the filtering of \cref{eq:2Dfilter_freq}, for any vertex $\seqprn{i_1,i_2}\in V_1\times V_2$, an output value $\app{f_\OUT}{i_1,i_2}$ is a linear combination of input values $\setprnsep{\app{f_\IN}{j_1,j_2}}{\seqprn{j_1,j_2}\in\app{\mathcal{N}}{i_1,i_2}}$.
\end{corollary}

We can easily prove this corollary from \cref{prop:2Dfilter_vertex} and~\citep[Lemma 5.2]{Hammond2011}.
\Cref{cor:2Dfilter_locality} shows that the 2-D graph spectral filter defined by a kernel $\vhat{h}_{S_1S_2}$ propagates a signal element $\app{f_\IN}{i_1,i_2}$ only in the local neighborhood $\app{\mathcal{N}}{i_1,i_2}$.

\subsection{Optimization filtering}\label{subsec:opt_filtering}
Given a noisy observation of graph signals, consider an estimation of its original signal.
Such problems are sometimes attributed to optimizations, referred to as \newword{an optimization filtering} in this study.
Suppose a graph signal $y$ is observed.
In an optimization filtering framework, we estimate the original graph signal by
\begin{equation}
  x_\opt\in\argmin_{x}\norm[p]{x-y}^p+\app{\mathcal{S}}{x}
\end{equation}
with some function $\mathcal{S}$ that represents how strongly the signal changes along the graph, similar to total quadratic variation.
The term $\app{\mathcal{S}}{x}$ smooths the estimator on the graph, and the other term $\norm[p]{x-y}^p$ brings it close to the observation.
Several optimization filtering frameworks are discussed in~\citep{Zhou2004,Zhou2005}.

On a Cartesian product graph, our 2-D GFT suggests separately handling two signal fluctuations along the factor graphs, even though existing optimization filtering methods do not do so.
In this subsection, a multi-dimensional version of several optimization filtering frameworks will be proposed.

We multi-dimensionalize \newword{an extended basic energy model (EBEM)}~\citep{Grady2010}, a subclass of an optimization filtering in this section.
Let $G=\seqprn{V,E,w}$ be an undirected weighted graph with vertex set $V=\setprn{0,\ldots,N-1}$ and $y$ an observed signal on the graph.
The EBEM designs an ``energy'' of a graph signal $\funcdoms{x}{V}{\setR}$ on the graph as
\begin{equation}\label{eq:EBEM}
  \app{\mathcal{E}^G_\gamma}{x}
  =\norm[p]{x-y}^p+\frac{\gamma}{2}\sum_{i=0}^{N-1}\sum_{j=0}^{N-1}\app{w}{i,j}\absprn{\app{x}{i}-\app{x}{j}}^q
\end{equation}
with a regularization parameter $\gamma$, and estimates the true signal by the minimizer $x_\opt\in\argmin_{x}\app{\mathcal{E}^G_\gamma}{x}$.

Next, we will discuss the EBEM on a Cartesian product graph.
Let $G_n=\seqprn{V_n,E_n,w_n}$ be an undirected weighted graph with vertex set $V_n=\setprn{0,\ldots,N_n-1}$ for $n=1,2$ and $y$ an observed signal on the product graph $G_1\cprod G_2$.
The EBEM energy of a graph signal $\funcdoms{x}{V_1\times V_2}{\setR}$ on the product graph is given by
\begin{align}
  &\app{\mathcal{E}^{G_1\cprod G_2}_\gamma}{x}
  \\
  &=\norm[p]{x-y}^p+\frac{\gamma}{2}\sum_{i_1}\sum_{j_1}\app{w_1}{i_1,j_1}\norm[q]{\app{x}{i_1,\placeholder}-\app{x}{j_1,\placeholder}}^q
  \\
  &\hphantom{=}+\frac{\gamma}{2}\sum_{i_2}\sum_{j_2}\app{w_2}{i_2,j_2}\norm[q]{\app{x}{\placeholder,i_2}-\app{x}{\placeholder,j_2}}^q
\end{align}
in which parameters $\gamma$ and $q$ are isotropic along $G_1$ and $G_2$.
It is a natural extension to make the parameters anisotropic on factor graphs.

\begin{definition}[Two-dimensional extended basic energy model]
  Suppose a graph signal $\funcdoms{y}{V_1\times V_2}{\setR}$ on a Cartesian product graph $G_1\cprod G_2$ is observed.
  A two-dimensional extended basic energy model (2-D EBEM) estimates the original signal by a minimizer $\funcdoms{x}{V_1\times V_2}{\setR}$ of an energy
  \begin{align}
    &\app{\mathcal{E}^{G_1G_2}_{\gamma_1\gamma_2}}{x}
    \\
    &=\norm[p]{x-y}^p
    +\frac{\gamma_1}{2}\sum_{i_1=0}^{N_1-1}\sum_{j_1=0}^{N_1-1}\app{w_1}{i_1,j_1}\norm[q_1]{\app{x}{i_1,\placeholder}-\app{x}{j_1,\placeholder}}^{q_1}
    \\\label{eq:2-DEBEM}
    &\hphantom{=}+\frac{\gamma_2}{2}\sum_{i_2=0}^{N_2-1}\sum_{j_2=0}^{N_2-1}\app{w_2}{i_2,j_2}\norm[q_2]{\app{x}{\placeholder,i_2}-\app{x}{\placeholder,j_2}}^{q_2}.
  \end{align}
\end{definition}

In a 2-D EBEM, controlling regularization weights $\gamma_1$ and $\gamma_2$, or regularization dimensions $q_1$ and $q_2$ separately, we can design a factor-graph-anisotropic energy.
In particular, anisotropy between dimensional parameters is essential, while anisotropy between weight parameters can be attributed to the weight functions of graphs.

\section{Stationarity}\label{sec:stationarity}
This section explains that the proposed multi-dimensional graph spectral filtering framework drives the development of new stationarities of multi-dimensional random graph signals.

First, we refer to the stationarity of 1-D random graph signals.
Let $G$ be an undirected weighted graph with vertex set $V=\setprn{0,\ldots,N-1}$.
A zero-mean random graph signal $x$ on $G$, i.e., a zero-mean random variable on $V$, is \newword{(weak) stationary} when $x$ is an output of some $\paren{N-1}$-degree polynomial filter applied to a white noise $z$ on $V$~\citep{Segarra2017}.
Note that $z$ is independent, identically distributed, and satisfies $\EV{\app{z}{i}}=0$ and $\Cov{\app{z}{i},\app{z}{j}}=\app{\delta}{i,j}$.
In other words, by putting $\vec{x}=\trsps{\seqprn{\app{x}{0}\:\cdots\:\app{x}{N-1}}}$ and $\vec{z}=\trsps{\seqprn{\app{z}{0}\:\cdots\:\app{z}{N-1}}}$, a zero-mean random graph signal $x$ on $G$ is said to be stationary if some $h_0,\ldots,h_{N-1}\in\setC$ satisfy
\begin{align}
  \vec{x}=\paren{\sum_{s=0}^{N-1}h_s\mat{L}^s}\vec{z},
\end{align}
where $\mat{L}$ is the Laplacian matrix of $G$.

The previous study~\citep{Segarra2017} pointed out that for a stationary graph signal, its covariance matrix and a Laplacian matrix of the graph are simultaneously diagonalizable; furthermore, when these two matrices are simultaneously diagonalizable and the eigenvalues of the Laplacian matrix are distinct, the graph signal is stationary.

In this study, we propose two stationarities for 2-D graph signals: factor-graph-wise stationarity and directional stationarity.
These 2-D stationarities naturally deduce their $n$-D versions for an arbitrary natural number $n$.
For a signal on a product of a cycle graph and any graph, corresponding concepts already exist~\citep{Loukas2017}.
Our stationarities generalize these concepts.

Using an analogy to the existing stationarity, \newword{factor-graph-wise stationary} signals are defined as outputs of 2-D graph polynomial filters applied to 2-D white noise.
For $n=1,2$, let $G_n$ be an undirected weighted graph with vertex set $V_n=\setprn{0,\ldots,N_n-1}$ whose graph Laplacian matrix $\mat{L}_n$ has ascending eigenvalues $\rawcurlybrace{\uid{\lambda_k}{n}}_{k=0,\ldots,N_n-1}$.
Let $z$ be white noise on $V_1\times V_2$, i.e., independent and identically distributed, $\EV{\app{z}{i_1,i_2}}=0$, and $\Cov{\app{z}{i_1,i_2},\app{z}{j_1,j_2}}=\app{\delta}{i_1,j_1}\app{\delta}{i_2,j_2}$.

\begin{definition}[Factor-graph-wise stationarity]
  A zero-mean random graph signal $x$ on a Cartesian product graph $G_1\cprod G_2$ is said to be factor-graph-wise stationary if some $h_{00},h_{01},\ldots,h_{\paren{N_1-1}\paren{N_2-1}}\in\setC$ satisfy
  \begin{align}\label{eq:fgwise_stationarity}
    \mat{X}=\sum_{s_1=0}^{N_1-1}\sum_{s_2=0}^{N_2-1}h_{s_1s_2}\mat{L}_1^{s_1}\mat{Z}\mat{L}_2^{s_2},
  \end{align}
  where $\mat{X}$ and $\mat{Z}$ are $N_1\times N_2$ matrices $\seqprn{\app{x}{i_1,i_2}}_{i_1,i_2}$ and $\seqprn{\app{z}{i_1,i_2}}_{i_1,i_2}$, respectively.
\end{definition}

We present two important theorems regarding factor-graph-wise stationary signals below.
Define the following matrices:
\begin{itemize}
  \item an $N_1\times N_1$ matrix $\Cov{\app{x}{\placeholder,i_2},\app{x}{\placeholder,j_2}}$ whose $\seqprn{i_1,j_1}$-th element is $\Cov{\app{x}{i_1,i_2},\app{x}{j_1,j_2}}$,
  \item an $N_2\times N_2$ matrix $\Cov{\app{x}{i_1,\placeholder},\app{x}{j_1,\placeholder}}$ whose $\seqprn{i_2,j_2}$-th element is $\Cov{\app{x}{i_1,i_2},\app{x}{j_1,j_2}}$, and
  \item an $N_1N_2\times N_1N_2$ matrix $\Cov{x}$ whose $\seqprn{N_2i_1+i_2,N_2j_1+j_2}$-th element is $\Cov{\app{x}{i_1,i_2},\app{x}{j_1,j_2}}$.
\end{itemize}
Let $\vhat{x}$ be a 2-D spectrum of $x$ obtained by the 2-D GFT.

\begin{theorem}\label{th:fgwise_simdiag_to_spectral_uncorrelatedness}
  For a zero-mean random graph signal $x$ on a Cartesian product graph $G_1\cprod G_2$, the following three conditions are equivalent:
  \begin{enumerate}
    \item matrices $\Cov{\app{x}{\placeholder,i_2},\app{x}{\placeholder,j_2}}$ and $\mat{L}_1$ are simultaneously diagonalizable for any $i_2,j_2=0,\ldots,N_2-1$, and matrices $\Cov{\app{x}{i_1,\placeholder},\app{x}{j_1,\placeholder}}$ and $\mat{L}_2$ are simultaneously diagonalizable for any $i_1,j_1=0,\ldots,N_1-1$,
    \item a covariance between spectra $\app{\vhat{x}}{\uid{\lambda_{k_1}}{1},\uid{\lambda_{k_2}}{2}}$ and $\app{\vhat{x}}{\uid{\lambda_{l_1}}{1},\uid{\lambda_{l_2}}{2}}$ is equal to zero at $k_1\neq l_1$ and $k_2\neq l_2$, and
    \item matrices $\Cov{x}$ and $\mat{L}_1\ksum\mat{L}_2$ are simultaneously diagonalizable.
  \end{enumerate}
\end{theorem}

\begin{theorem}\label{th:fgwise_stationarity_nscond}
  If a zero-mean random graph signal $x$ on a Cartesian product graph $G_1\cprod G_2$ is factor-graph-wise stationary, the following two statements hold:
  \begin{enumerate}
    \item matrices $\Cov{\app{x}{\placeholder,i_2},\app{x}{\placeholder,j_2}}$ and $\mat{L}_1$ are simultaneously diagonalizable for any $i_2,j_2=0,\ldots,N_2-1$, and
    \item matrices $\Cov{\app{x}{i_1,\placeholder},\app{x}{j_1,\placeholder}}$ and $\mat{L}_2$ are simultaneously diagonalizable for any $i_1,j_1=0,\ldots,N_1-1$.
  \end{enumerate}
  When eigenvalues of $\mat{L}_1$ are distinct and eigenvalues of $\mat{L}_2$ are distinct, the converse holds.
\end{theorem}

Proofs of \cref{th:fgwise_stationarity_nscond,th:fgwise_simdiag_to_spectral_uncorrelatedness} are shown in the appendix.

The theorems show us the following two facts.
First, when assuming distinct eigenvalues of each factor graph, the factor-graph-wise stationarity is equal to the uncorrelatedness between different spectral components.
Second, when assuming distinct eigenvalues of the product graph, the factor-graph-wise stationarity is also equal to the existing stationarity.
\Cref{fig:stationarities_relationship} shows the relationship between stationarities about 2-D signals (concepts in the bottom row appears later in this subsection).
Note that if the eigenvalues of the product graph are distinct, the eigenvalues of each factor graph are distinct as well.
In this case, the existing stationarity and the factor-graph-wise stationarity are equivalent.
However, even if the eigenvalues of each factor graph are distinct, the eigenvalues of the product graph are not always distinct.
In this case, the existing stationarity implies factor-graph-wise stationarity, but the reverse is not always true.

\begin{figure*}
  \centering
  \includegraphics[width=\linewidth]{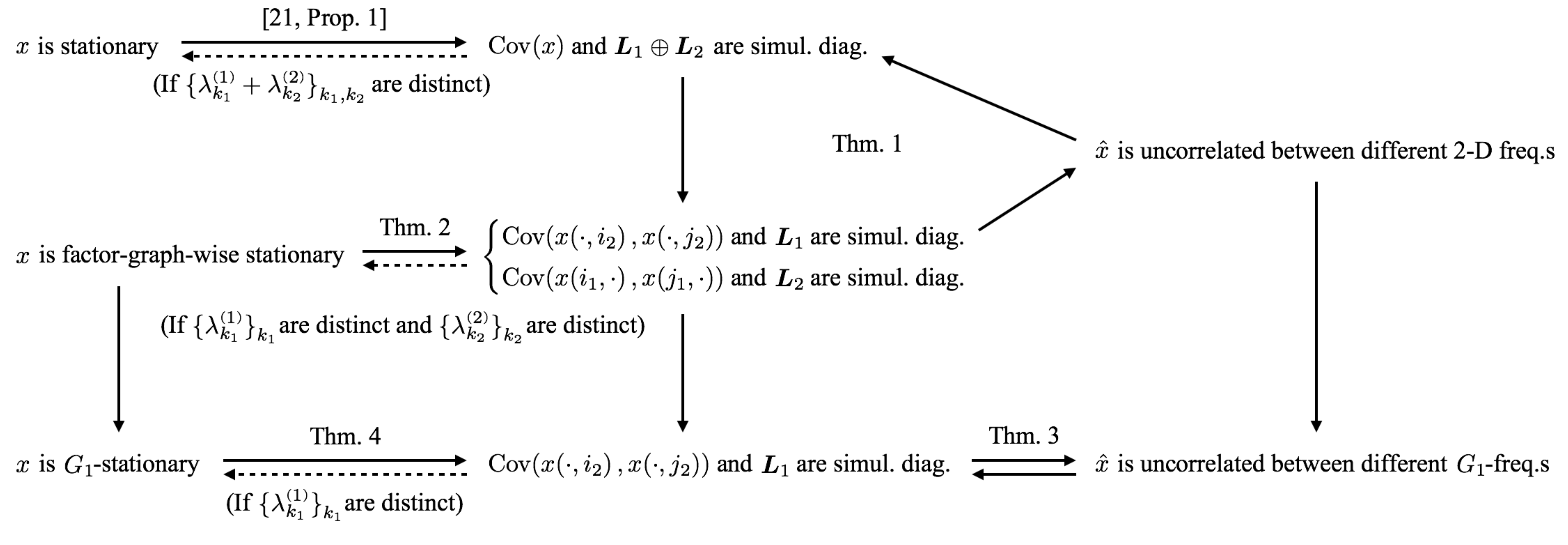}
  \caption{Relationships of various stationarities about a signal $x$ on a graph $G_1\cprod G_2$. The phrase ``simul. diag.'' stands for ``simultaneously diagonalizable.'' A solid arrow from $P$ to $Q$ means ``if $P$ then $Q$,'' and a dashed arrow from $P$ to $Q$ means ``if $P$ then $Q$ under some conditions.''}
  \label{fig:stationarities_relationship}
\end{figure*}

Next, we propose a directional stationarity that is a ``stationarity along one factor graph,'' whereas the factor-graph-wise stationarity is a ``stationarity along both factor graphs.''

\begin{definition}[Directional stationarity]
  A zero-mean random graph signal $x$ on a Cartesian product graph $G_1\cprod G_2$ is said to be $G_1$-stationary if some matrices $\mat{H}_0,\ldots,\mat{H}_{N_1-1}\in\setC^{N_2\times N_2}$ satisfy
  \begin{align}
    \mat{X}=\sum_{s_1=0}^{N_1-1}\mat{L}_1^{s_1}\mat{Z}\mat{H}_{s_1},\label{eq:G1_stationarity}
  \end{align}
  and is said to be $G_2$-stationary if some matrices $\mat{H}_0,\ldots,\mat{H}_{N_2-1}\in\setC^{N_1\times N_1}$ satisfy
  \begin{align}
    \mat{X}=\sum_{s_2=0}^{N_2-1}\mat{H}_{s_2}\mat{Z}\mat{L}_2^{s_2}.\label{eq:G2_stationarity}
  \end{align}
\end{definition}

In \cref{subsec:multivariate_stationarity}, the directional stationarity of 2-D graph signals deduces a stationarity of multivariate graph signals.
We present two important theorems regarding directionally stationary signals below.

\begin{theorem}\label{th:directional_simdiag_to_spectral_uncorrelatedness}
  \textbf{(A)} For a zero-mean random graph signal $x$ on a Cartesian product graph $G_1\cprod G_2$, the following two conditions are equivalent:
  \begin{enumerate}
    \item matrices $\Cov{\app{x}{\placeholder,i_2},\app{x}{\placeholder,j_2}}$ and $\mat{L}_1$ are simultaneously diagonalizable for any $i_2,j_2=0,\ldots,N_2-1$, and
    \item a covariance between spectra $\app{\vhat{x}}{\uid{\lambda_{k_1}}{1},\uid{\lambda_{k_2}}{2}}$ and $\app{\vhat{x}}{\uid{\lambda_{l_1}}{1},\uid{\lambda_{l_2}}{2}}$ is equal to zero at $k_1\neq l_1$.
  \end{enumerate}

  \textbf{(B)} For a zero-mean random graph signal $x$ on a Cartesian product graph $G_1\cprod G_2$, the following two conditions are equivalent:
  \begin{enumerate}
    \item matrices $\Cov{\app{x}{i_1,\placeholder},\app{x}{j_1,\placeholder}}$ and $\mat{L}_2$ are simultaneously diagonalizable for any $i_1,j_1=0,\ldots,N_1-1$, and
    \item a covariance between spectra $\app{\vhat{x}}{\uid{\lambda_{k_1}}{1},\uid{\lambda_{k_2}}{2}}$ and $\app{\vhat{x}}{\uid{\lambda_{l_1}}{1},\uid{\lambda_{l_2}}{2}}$ is equal to zero at $k_2\neq l_2$.
  \end{enumerate}
\end{theorem}

\begin{theorem}\label{th:directional_stationarity_nscond}
  \textbf{(A)} If a zero-mean random graph signal $x$ on a Cartesian product graph $G_1\cprod G_2$ is $G_1$-stationary, matrices $\Cov{\app{x}{\placeholder,i_2},\app{x}{\placeholder,j_2}}$ and $\mat{L}_1$ are simultaneously diagonalizable for any $i_2,j_2=0,\ldots,N_2-1$.
  When eigenvalues of $\mat{L}_1$ are distinct, the converse holds.

  \textbf{(B)} If a zero-mean random graph signal $x$ on a Cartesian product graph $G_1\cprod G_2$ is $G_2$-stationary, matrices $\Cov{\app{x}{i_1,\placeholder},\app{x}{j_1,\placeholder}}$ and $\mat{L}_2$ are simultaneously diagonalizable for any $i_1,j_1=0,\ldots,N_1-1$.
  When eigenvalues of $\mat{L}_2$ are distinct, the converse holds.
\end{theorem}

Proofs of \cref{th:directional_simdiag_to_spectral_uncorrelatedness,th:directional_stationarity_nscond} are shown in the appendix.

\section{Multivariate graph signal processing}\label{sec:multivariate_gsp}

\subsection{Multivariate graph signals}
It is natural to extend graph signals to their multivariate version; however, such signals have rarely been discussed.
Let $G$ be an undirected weighted graph with vertex set $V=\setprn{0,\ldots,N-1}$.
A vector-valued function $\funcdoms{\vec{f}}{V}{\setR^p}$ is called a $p$-variate graph signal on $G$.
Multivariate graph signals often appear in the real world.
For example, RGB signals on pixels are 3-variate signals on the grid graph, and 3-D acceleration measured by scattered seismometers are 3-variate signals on the observation network graph.

The GFT of multivariate graph signals is defined as a variable-wise GFT.
Suppose that the graph Laplacian $\mat{L}$ of $G$ has ascending eigenvalues $\rawcurlybrace{\lambda_k}_{k=0,\ldots,N-1}$ and the corresponding eigenfunctions $\rawcurlybrace{u_k}_{k=0,\ldots,N-1}$ on $V$.

\begin{definition}[Multivariate graph Fourier transform]
  A $p$-variate GFT of a $p$-variate signal $\funcdoms{\vec{f}}{V}{\setR^p}$ on a graph $G$ is a $p$-variate spectrum $\funcdoms{\vhat{\vec{f}}}{\app{\sigma}{\mat{L}}}{\setC^p}$ defined by
  \begin{align}
    \app{\vhat{\vec{f}}}{\lambda_k}
    =\sum_{i=0}^{N-1}\app{\vec{f}}{i}\wideconj{\app{u_k}{i}}
  \end{align}
  for $k=0,\ldots,N-1$, and its inverse is given by
  \begin{align}
    \app{\vec{f}}{i}
    =\sum_{k=0}^{N-1}\app{\vhat{\vec{f}}}{\lambda_k}\app{u_k}{i}
  \end{align}
  for $i=0,\ldots,N-1$.
\end{definition}

Note that the multivariate GFT is represented as a matrix-matrix multiplication.
Let $\funcdoms{\vec{f}}{V}{\setR^p}$ be a $p$-variate graph signal.
By using $N\times p$ matrices $\mat{F}=\trsps{\seqprn{\app{\vec{f}}{0}\:\cdots\:\app{\vec{f}}{N-1}}}$ and $\vhat{\mat{F}}=\trsps{\seqprn{\app{\vhat{\vec{f}}}{\lambda_0}\:\cdots\:\app{\vhat{\vec{f}}}{\lambda_{N-1}}}}$, the $p$-variate GFT applied to $\vec{f}$ is expressed as $\vhat{\mat{F}}=\adjoint{\mat{U}}\mat{F}$, where $\mat{U}$ is an $N\times N$ unitary matrix with $\seqprn{i,k}$-th element $\app{u_k}{i}$.
Then, its inverse is given by $\mat{F}=\mat{U}\vhat{\mat{F}}$.

Multivariate signals on a graph can be regarded as univariate signals on a product graph of the graph and an edgeless graph.
Let $\bar{K}_p$ be an edgeless graph with vertex set $\setprn{0,\ldots,p-1}$.
A $p$-variate signal $\funcdoms{\vec{f}}{V}{\setR^p}$ is equated to a univariate signal $\funcdoms{g}{V\times\setprn{0,\ldots,p-1}}{\setR}$ on $G\cprod\bar{K}_p$ satisfying that $\app{g}{i,a}$ is the $a$-th variable of $\app{\vec{f}}{i}$ for any $i=0,\ldots,N-1$ and any $a=0,\ldots,p-1$.
The product graph $G\cprod\bar{K}_p$ is $p$ independent copies of $G$.

Then, the multivariate GFT of multivariate signals is equal to the 2-D GFT of the corresponding 2-D univariate signals.
The graph Laplacian of $\bar{K}_p$ is a zero matrix so that one of its eigenvectors are the standard basis in $\setC^p$.
Therefore, by putting $\mat{G}=\rawparen{\app{g}{i,a}}_{i,a}$, the 2-D spectrum of $g$ obtained by the 2-D GFT has a matrix representation $\vhat{\mat{G}}=\adjoint{\mat{U}}\mat{G}\matid[p]=\adjoint{\mat{U}}\mat{G}$.
Considering $\mat{G}=\trsps{\seqprn{\app{\vec{f}}{0}\:\cdots\:\app{\vec{f}}{N-1}}}$, the multivariate GFT and the 2-D GFT are equal.
Note that because the graph Laplacian of $\bar{K}_p$ only has an eigenvalue of zero, the 2-D spectrum of $g$ is multi-valued at all frequencies.
Matrix representations of GFTs conceal their multi-valuedness under non-distinct frequencies.

\subsection{Stationarity of multivariate graph signals}\label{subsec:multivariate_stationarity}
This subsection will propose the stationarity of multivariate random signals on graphs, which extends existing stationarities for univariate graph signals in~\citep{Segarra2017,Perraudin2017}.
For a $p$-variate random graph signal $\vec{x}=\rawparen{x_a}_{a=0,\ldots,p-1}$ on $G$, denote an $N\times p$ matrix whose $\seqprn{i,a}$-element is $\app{x_a}{i}$ by $\mat{X}$.
Let $\mat{Z}$ be an $N\times p$ white noise matrix.

\begin{definition}[Stationarity of multivariate graph signals]\label{def:multivariate_stationarity}
  A zero-mean $p$-variate random signal $\vec{x}$ on a graph $G$ is said to be stationary if some matrices $\mat{H}_0,\ldots,\mat{H}_{N-1}\in\setC^{p\times p}$ satisfy
  \begin{align}\label{eq:multivariate_stationarity}
    \mat{X}=\sum_{s=0}^{N-1}\mat{L}^s\mat{Z}\mat{H}_s.
  \end{align}
\end{definition}

Considering multivariate graph signals as 2-D univariate graph signals as mentioned above, the stationarity of multivariate graph signals is equal to the directional stationarity of the corresponding 2-D graph signals.
\Cref{th:directional_simdiag_to_spectral_uncorrelatedness,th:directional_stationarity_nscond} can also be translated to the following corollaries about multivariate stationary signals.
Let $\Cov{x_a,x_b}$ be an $N\times N$ matrix whose $\seqprn{i,j}$-th element is $\Cov{\app{x_a}{i},\app{x_b}{j}}$ for $i,j=0,\ldots,N-1$ and $\vhat{x}_a$ the $a$-th variable of the spectrum $\vhat{\vec{x}}$.

\begin{corollary}\label{cor:multivariate_simdiag_to_spectral_uncorrelatedness}
  For a zero-mean $p$-variate random signal $\vec{x}$ on a graph $G$, the following two conditions are equivalent:
  \begin{enumerate}
    \item matrices $\Cov{x_a,x_b}$ and $\mat{L}$ are simultaneously diagonalizable for any $a,b=0,\ldots,p-1$, and
    \item a covariance between spectra $\app{\vhat{x}_a}{\lambda_k}$ and $\app{\vhat{x}_b}{\lambda_l}$ is equal to zero at $k\neq l$.
  \end{enumerate}
\end{corollary}

\begin{corollary}\label{cor:multivariate_stationarity_nscond}
  If a zero-mean $p$-variate random signal $\vec{x}$ on $G$ is stationary, matrices $\Cov{x_a,x_b}$ and $\mat{L}$ are simultaneously diagonalizable for any $a,b=0,\ldots,p-1$.
  When eigenvalues of $\mat{L}$ are distinct, the converse holds.
\end{corollary}

The proposed multivariate stationarity is consistent with the stationarity of multivariate time signals.
Suppose a zero-mean $p$-variate random signal $\vec{x}$ on an $N$-cycle graph is stationary, and then \cref{cor:multivariate_stationarity_nscond} indicates that the covariance $\Cov{x_a,x_b}$ and the graph Laplacian $\mat{L}$ are simultaneously diagonalizable for any $a,b=0,\ldots,p-1$.
Because $\mat{L}$ is diagonalizable with the discrete Fourier matrix, the matrix $\Cov{x_a,x_b}$ should be circular, i.e., satisfies $\Cov{\app{x_a}{i},\app{x_b}{j}}=\Cov{\app{x_a}{i-j},\app{x_b}{0}}$ for any $i,j=0,\ldots,N-1$ and any $a,b=0,\ldots,p-1$.
Note that $\app{x_a}{i-j}$ is $x_a$ applied to the remainder of $i-j$ modulo $N$.
Therefore, regarding the graph signal $\vec{x}$ as a $p$-variate time signal of period $N$, its autocovariance matrix $\rawparen{\Cov{\app{x_a}{i},\app{x_b}{j}}}_{a,b=0,\ldots,p-1}$ is shift-invariant, which means that the signal is stationary as a temporal signal as well.

\section{Conclusion}\label{sec:conclusion}
This study has proposed an MGFT that retains the dimensional information of multi-dimensional graph signals.
The proposed transform has provided multi-dimensional spectral filtering, multi-dimensional optimization filtering, factor-graph-wise stationarity, and directional stationarity.
By considering multivariate graph signals as 2-D univariate graph signals, this study has proposed the multivariate GFT and stationarity.

Note that the proposed multi-dimensional GSP methodologies are not applicable to signals on a product graph with unknown factor graphs, or to signals on a nearly product graph (even small perturbation destroys Cartesian product structure~\citep{Imrich1996}).
Preliminary decomposition or approximation of graphs with product graphs~\citep{Imrich2008,Hagauer1999,Hellmuth2013} may solve the problem.

Further work is needed to build upon the findings of this study.
One future work is to clarify whether the proposed stationarities exist in practical graph signals or not.
For 1-D univariate graph signals, numerical experiments indicated that the well-known USPS dataset was almost stationary~\citep{Perraudin2017}.
Further numerical experiments may show that a practical multi-dimensional graph signal is almost factor-wise or directional stationary, and that a practical multivariate graph signal is almost stationary.
Another future work is to prove the importance of stationarities for graph signals.
For 1-D univariate graph signals, the stationarity enables us to estimate the power spectral densities of the signals~\citep{Segarra2017} and to construct Wiener filters on graphs~\citep{Perraudin2017}.
We expect that the assumption of the proposed stationarities will provide new GSP methodologies.

\appendix
Here are proofs of theorems in \cref{sec:stationarity}.
For $n=1,2$, let $G_n$ be an undirected weighted graph with vertex set $V_n=\setprn{0,\ldots,N_n-1}$, and suppose that its graph Laplacian $\mat{L}_n$ can be decomposed as $\mat{L}_n=\mat{U}_n\mat{\Lambda}_n\adjoint{\mat{U}_n}$ with a diagonal matrix $\mat{\Lambda}_n$ whose diagonal elements are $\uid{\lambda_0}{n},\ldots,\uid{\lambda_{N-1}}{n}$ and a unitary matrix $\mat{U}_n=\rawparen{\app{\uid{u_k}{n}}{i}}_{i,k}=\rawparen{\uid{u_{ik}}{n}}_{i,k}$.
For a graph signal $\funcdoms{x}{V_1\times V_2}{\setR}$ on $G_1\cprod G_2$, denote $\app{x}{i_1,i_2}$ by $x_{i_1i_2}$ and define an $N_1\times N_2$ matrix $\mat{X}=\rawparen{x_{i_1i_2}}_{i_1,i_2}$.
For the spectrum $\vhat{x}$ of $x$ obtained by the 2-D GFT, denote $\app{\vhat{x}}{\uid{\lambda_{k_1}}{1},\uid{\lambda_{k_2}}{2}}$ by $\vhat{x}_{k_1k_2}$ and define an $N_1\times N_2$ matrix $\vhat{\mat{X}}=\rawparen{\vhat{x}_{k_1k_2}}_{k_1,k_2}$.
For a white noise function $\funcdoms{z}{V_1\times V_2}{\setR}$ on $G_1\cprod G_2$ and its 2-D spectrum $\vhat{z}$, define $z_{i_1i_2}$, $\mat{Z}$, $\vhat{z}_{k_1k_2}$, and $\vhat{\mat{Z}}$ in the same manner.
For any variable with two indices like $x_{i_1i_2}$, denote the $i_1$-th row vector $\rawparen{x_{i_10}\:\cdots\:x_{i_1\paren{N_2-1}}}$ by $\vec{x}_{i_1\nid}$ and the $i_2$-th column vector $\rawparen{x_{0i_2}\:\cdots\:x_{\paren{N_1-1}i_2}}^\top$ by $\vec{x}_{\nid i_2}$.
Denote a diagonal matrix whose diagonal elements are $a_1,\ldots,a_n$ by $\app{\diag}{a_1,\ldots,a_n}$, and with $\vec{a}=\rawparen{a_i}_{i=1,\ldots,n}$, denote the diagonal matrix by $\diag\vec{a}$ also.

\subsection{Proofs of \cref{th:fgwise_simdiag_to_spectral_uncorrelatedness,th:fgwise_stationarity_nscond} about factor-graph-wise stationarity}
First we prove \cref{th:fgwise_simdiag_to_spectral_uncorrelatedness}.

\begin{proof}[Proof of \cref{th:fgwise_simdiag_to_spectral_uncorrelatedness}]
  We show that 1) and 2) are equivalent, and that 2) and 3) are equivalent.

  [from 1) to 2)]
  Suppose condition 1).
  Due to the first simultaneous diagonalizability, some vector $\vec{\alpha}_{i_2j_2}$ satisfies
  \begin{align}
    \Cov{\app{x}{\placeholder,i_2},\app{x}{\placeholder,j_2}}
    =\mat{U}_1\paren{\diag\vec{\alpha}_{i_2j_2}}\adjoint{\mat{U}_1}
  \end{align}
  for $i_2,j_2=0,\ldots,N_2-1$, and due to the second, some vector $\vec{\beta}_{i_1j_1}$ satisfies
  \begin{align}
    \Cov{\app{x}{i_1,\placeholder},\app{x}{j_1,\placeholder}}
    =\mat{U}_2\paren{\diag\vec{\beta}_{i_1j_1}}\adjoint{\mat{U}_2}
  \end{align}
  for $i_1,j_1=0,\ldots,N_1-1$.
  Denote the $k_1$-th element of $\vec{\alpha}_{i_2j_2}$ by $\alpha_{i_2j_2,k_1}$ and the $k_2$-th element of $\vec{\beta}_{i_1j_1}$ by $\beta_{i_1j_1,k_2}$, and put
  \begin{align}
    \gamma_{k_1k_2}=\sum_{i_1}\sum_{j_1}\wideconj{\uid{u_{i_1k_1}}{1}}\uid{u_{j_1k_1}}{1}\beta_{i_1j_1,k_2}
  \end{align}
  for any $k_1=0,\ldots,N_1-1$ and $k_2=0,\ldots,N_2-1$.
  Because a covariance $\Cov{\app{x}{i_1,i_2}, \app{x}{j_1,j_2}}$ has two expressions like $\sum_{k_1}\uid{u_{i_1k_1}}{1}\wideconj{\uid{u_{j_1k_1}}{1}}\alpha_{i_2j_2,k_1}$ and $\sum_{k_2}\uid{u_{i_2k_2}}{2}\wideconj{\uid{u_{j_2k_2}}{2}}\beta_{i_1j_1,k_2}$, an equation
  \begin{align}
    \alpha_{i_2j_2,k_1}
    &=\sum_{i_1}\sum_{j_1}\wideconj{\uid{u_{i_1k_1}}{1}}\uid{u_{j_1k_1}}{1}\paren{\sum_{l_1}\uid{u_{i_1l_1}}{1}\wideconj{\uid{u_{j_1l_1}}{1}}\alpha_{i_2j_2,l_1}}
    \\
    &=\sum_{i_1}\sum_{j_1}\wideconj{\uid{u_{i_1k_1}}{1}}\uid{u_{j_1k_1}}{1}\paren{\sum_{k_2}\uid{u_{i_2k_2}}{2}\wideconj{\uid{u_{j_2k_2}}{2}}\beta_{i_1j_1,k_2}}
    \\
    &=\sum_{k_2}\uid{u_{i_2k_2}}{2}\wideconj{\uid{u_{j_2k_2}}{2}}\gamma_{k_1k_2}
  \end{align}
  holds.
  Therefore, a spectral covariance of $x$ is given by
  \begin{align}
    &\Cov{\vhat{\vec{x}}_{\nid k_2},\vhat{\vec{x}}_{\nid l_2}}
    =\EV{\adjoint{\mat{U}_1}\mat{X}\wideconj{\uid{\vec{u}_{\nid k_2}}{2}}\vtrsps{\paren{\vuid{\vec{u}_{\nid l_2}}{2}}}\vtrsps{\vec{X}}\mat{U}_1}
    \\
    &=\adjoint{\mat{U}_1}\paren{\sum_{i_2}\sum_{j_2}\wideconj{\uid{u_{i_2k_2}}{2}}\uid{u_{j_2l_2}}{2}\EV{\vec{x}_{\nid i_2}\trsps{\vec{x}_{\nid j_2}}}}\mat{U}_1
    \\
    &=\adjoint{\mat{U}_1}\paren{\sum_{i_2}\sum_{j_2}\wideconj{\uid{u_{i_2k_2}}{2}}\uid{u_{j_2l_2}}{2}\mat{U}_1\paren{\diag\vec{\alpha}_{i_2j_2}}\adjoint{\mat{U}_1}}\mat{U}_1
    \\
    &=\sum_{i_2}\sum_{j_2}\wideconj{\uid{u_{i_2k_2}}{2}}\uid{u_{j_2l_2}}{2}\diag\vec{\alpha}_{i_2j_2}
    =\delta_{k_2l_2}\diag\vec{\gamma}_{\nid k_2}
  \end{align}
  and the condition 2) holds.

  [from 2) to 1)]
  Suppose condition 2).
  Then some $\gamma_{00},\ldots,\gamma_{\paren{N_1-1}\paren{N_2-1}}$ exist and satisfy
  \begin{align}
    \Cov{\vhat{x}_{k_1k_2},\vhat{x}_{l_1l_2}}
    =\delta_{k_1l_1}\delta_{k_2l_2}\gamma_{k_1k_2},\label{eq:spectral_uncorrelatedness}
  \end{align}
  and by putting $\alpha_{i_2j_2,k_1}=\sum_{k_2}\uid{u_{i_2k_2}}{2}\wideconj{\uid{u_{j_2k_2}}{2}}\gamma_{k_1k_2}$, $\mat{U}_1$ diagonalizes $\Cov{\app{x}{\placeholder,i_2},\app{x}{\placeholder,j_2}}$ into $\app{\diag}{\alpha_{i_2j_2,0},\ldots,\alpha_{i_2j_2,N_1-1}}$;
  by putting $\beta_{i_1j_1,k_2}=\sum_{k_1}\uid{u_{i_1k_1}}{1}\wideconj{\uid{u_{j_1k_1}}{1}}\gamma_{k_1k_2}$, $\mat{U}_2$ diagonalizes $\Cov{\app{x}{i_1,\placeholder},\app{x}{j_1,\placeholder}}$ into $\app{\diag}{\beta_{i_1j_1,0},\ldots,\beta_{i_1j_1,N_2-1}}$.
  Now the condition 1) holds.

  [2) and 3) are equal]
  The two conditions are obviously equivalent because of a relation of
  \begin{align}
    \Cov{\vhat{x}_{k_1k_2},\vhat{x}_{l_1l_2}}
    =
    \adjoint{\paren{\vuid{\vec{u}_{\nid k_1}}{1}\kprod\vuid{\vec{u}_{\nid k_2}}{2}}}
    \Cov{x}
    \paren{\vuid{\vec{u}_{\nid l_1}}{1}\kprod\vuid{\vec{u}_{\nid l_2}}{2}}.
  \end{align}
\end{proof}

Second, we prove \cref{th:fgwise_stationarity_nscond}.
Prepare a spectral representation of a 2-D graph signal filter in \cref{eq:fgwise_stationarity} like
\begin{align}
  \vhat{\mat{X}}
  &=\sum_{s_1=0}^{N_1-1}\sum_{s_2=0}^{N_2-1}h_{s_1s_2}\mat{\Lambda}_1^{s_1}\vhat{\mat{Z}}\mat{\Lambda}_2^{s_2}
  =\paren{\mat{\Psi}_1\mat{H}\trsps{\mat{\Psi}_2}}\hprod\vhat{\mat{Z}}\label{eq:fgwise_stationarity_spec},
\end{align}
where $\mat{H}$ is an $N_1\times N_2$ matrix with $\seqprn{s_1,s_2}$-th element $h_{s_1s_2}$ for $s_1=0,\ldots,N_1-1$ and $s_2=0,\ldots,N_2-1$, and $\mat{\Psi}_n$ is a Vandermonde matrix given by
\begin{align}
  \mat{\Psi}_n
  =
  \begin{pmatrix}
    \paren{\uid{\lambda_0}{n}}^0&\cdots&\paren{\uid{\lambda_0}{n}}^{N_n-1}
    \\
    \vdots&\ddots&\vdots
    \\
    \paren{\uid{\lambda_{N_n-1}}{n}}^0&\cdots&\paren{\uid{\lambda_{N_n-1}}{n}}^{N_n-1}
  \end{pmatrix}
\end{align}
for $n=1,2$.
The operator $\hprod$ indicates an element-wise product.

\begin{lemma}\label{lem:fgwise_stationary_spcov}
  When a signal $x$ on the graph $G_1\cprod G_2$ is factor-graph-wise stationary and represented as \cref{eq:fgwise_stationarity}, its spectral covariance is given by
  \begin{align}
    \Cov{\vhat{x}_{k_1k_2},\vhat{x}_{l_1l_2}}
    =\delta_{k_1l_1}\delta_{k_2l_2}\absprn{\vtilde{h}_{k_1k_2}}^2
  \end{align}
  for any $k_1,l_1=0,\ldots,N_1-1$ and $k_2,l_2=0,\ldots,N_2-1$, where $\vtilde{h}_{k_1k_2}$ is the $\seqprn{k_1,k_2}$-th element of the matrix $\mat{\Psi}_1\mat{H}\trsps{\mat{\Psi}_2}$.
\end{lemma}

\begin{proof}
  The spectral covariance of a white signal $z$ on $G_1\cprod G_2$ is given by
  \begin{align}
    &\Cov{\vhat{z}_{k_1k_2},\vhat{z}_{l_1l_2}}
    \\
    &=\sum_{i_1}\sum_{i_2}\sum_{j_1}\sum_{j_2}\wideconj{\uid{u_{i_1k_1}}{1}\uid{u_{i_2k_2}}{2}}\uid{u_{j_1l_1}}{1}\uid{u_{j_2l_2}}{2}\EV{z_{i_1i_2}z_{j_1j_2}}
    \\
    &=\sum_{i_1}\sum_{i_2}\sum_{j_1}\sum_{j_2}\wideconj{\uid{u_{i_1k_1}}{1}\uid{u_{i_2k_2}}{2}}\uid{u_{j_1l_1}}{1}\uid{u_{j_2l_2}}{2}\delta_{i_1j_1}\delta_{i_2j_2}
    \\
    &=\delta_{k_1l_1}\delta_{k_2l_2}.
  \end{align}
  Therefore, according to \cref{eq:fgwise_stationarity_spec}, the spectral covariance of $x$ is given by
  \begin{align}
    \Cov{\vhat{x}_{k_1k_2},\vhat{x}_{l_1l_2}}
    &=\Cov{\vtilde{h}_{k_1k_2}\vhat{z}_{k_1k_2},\vtilde{h}_{l_1l_2}\vhat{z}_{l_1l_2}}
    \\
    &=\delta_{k_1l_1}\delta_{k_2l_2}\absprn{\vtilde{h}_{k_1k_2}}^2.
  \end{align}
\end{proof}

\begin{proof}[Proof of \cref{th:fgwise_stationarity_nscond}]
  [from factor-graph-wise stationarity to simultaneous diagonalizability]
  Assuming signal $x$ is factor-graph-wise stationary, the 2-D spectrum $\vhat{x}$ is uncorrelated between different frequencies due to \cref{lem:fgwise_stationary_spcov}.
  Then, according to \cref{th:fgwise_simdiag_to_spectral_uncorrelatedness}, a matrix $\Cov{\app{x}{\placeholder,i_2},\app{x}{\placeholder,j_2}}$ is simultaneously diagonalizable with $\mat{L}_1$ for any $i_2,j_2=0,\ldots,N_2-1$ and $\Cov{\app{x}{i_1,\placeholder},\app{x}{j_1,\placeholder}}$ is simultaneously diagonalizable with $\mat{L}_2$ for any $i_1,j_1=0,\ldots,N_1-1$.

  [from simultaneous diagonalizability to factor-graph-wise stationarity]
  Suppose that a zero-mean signal $x$ on the graph satisfies the two simultaneous diagonalizabilities such that the eigenvalues $\uid{\lambda_0}{1},\ldots,\uid{\lambda_{N_1-1}}{1}$ are distinct, and that the eigenvalues $\uid{\lambda_0}{2},\ldots,\uid{\lambda_{N_2-1}}{2}$ are distinct.
  Using $\gamma_{00},\ldots,\gamma_{\paren{N_1-1}\paren{N_2-1}}$ satisfying \cref{eq:spectral_uncorrelatedness} from \cref{th:fgwise_simdiag_to_spectral_uncorrelatedness}, define a matrix
  \begin{align}
    \mat{H}=\mat{\Psi}_1^{-1}\sqrt{\mat{\Gamma}}\mat{\Psi}_2^{-1}
  \end{align}
  where $\sqrt{\mat{\Gamma}}$ is an $N_1\times N_2$ matrix with $\seqprn{k_1,k_2}$-th element $\sqrt{\gamma_{k_1k_1}}$.
  The matrix $\mat{H}$ exists because $\gamma_{k_1k_2}=\Var{\vhat{x}_{k_1k_2}}$ is nonnegative and because Vandermonde matrices $\mat{\Psi}_1$ and $\mat{\Psi}_2$ are invertible due to the distinctness of the eigenvalues.
  Then, the signal $x$ can be represented like \cref{eq:fgwise_stationarity} where $h_{s_1s_2}$ is the $\seqprn{s_1,s_2}$-th element of the matrix $\mat{H}$ for $s_1=0,\ldots,N_1-1$ and $s_2=0,\ldots,N_2-1$, and is factor-graph-wise stationary.
\end{proof}

\subsection{Proofs of \cref{th:directional_simdiag_to_spectral_uncorrelatedness,th:directional_stationarity_nscond} about directional stationarity}
First we show that \cref{th:directional_simdiag_to_spectral_uncorrelatedness}(A) holds.
The part (B) can be proved in the same way.

\begin{proof}[Proof of \cref{th:directional_simdiag_to_spectral_uncorrelatedness}(A)]
  [from 1) to 2)]
  Suppose the condition 1).
  Due to the simultaneous diagonalizability, some vector $\vec{\alpha}_{i_2j_2}$ satisfies
  \begin{align}
    \Cov{\app{x}{\placeholder,i_2},\app{x}{\placeholder,j_2}}
    =\mat{U}_1\paren{\diag\vec{\alpha}_{i_2j_2}}\adjoint{\mat{U}_1}
  \end{align}
  for $i_2,j_2=0,\ldots,N_2-1$.
  Denote the $k_1$-th element of $\vec{\alpha}_{i_2j_2}$ by $\alpha_{i_2j_2,k_1}$, where the index $k_1$ starts from zero.
  Then a spectral covariance of $x$ is given by
  \begin{align}
    &\Cov{\vhat{\vec{x}}_{\nid k_2},\vhat{\vec{x}}_{\nid l_2}}
    =\EV{\adjoint{\mat{U}_1}\mat{X}\wideconj{\uid{\vec{u}_{\nid k_2}}{2}}\vtrsps{\paren{\vuid{\vec{u}_{\nid l_2}}{2}}}\vtrsps{\vec{X}}\mat{U}_1}
    \\
    &=\adjoint{\mat{U}_1}\paren{\sum_{i_2}\sum_{j_2}\wideconj{\uid{u_{i_2k_2}}{2}}\uid{u_{j_2l_2}}{2}\EV{\vec{x}_{\nid i_2}\trsps{\vec{x}_{\nid j_2}}}}\mat{U}_1
    \\
    &=\adjoint{\mat{U}_1}\paren{\sum_{i_2}\sum_{j_2}\wideconj{\uid{u_{i_2k_2}}{2}}\uid{u_{j_2l_2}}{2}\mat{U}_1\paren{\diag\vec{\alpha}_{i_2j_2}}\adjoint{\mat{U}_1}}\mat{U}_1
    \\
    &=\app{\diag}{\sum_{i_2}\sum_{j_2}\wideconj{\uid{u_{i_2k_2}}{2}}\uid{u_{j_2l_2}}{2}\vec{\alpha}_{i_2j_2}}
  \end{align}
  and condition 2) holds.

  [from 2) to 1)]
  Suppose condition 2).
  Then, some vectors $\vec{\beta}_{00},\ldots,\vec{\beta}_{\paren{N_2-1}\paren{N_2-1}}$ exist and satisfy
  \begin{align}
    \Cov{\vhat{\vec{x}}_{\nid k_2},\vhat{\vec{x}}_{\nid l_2}}
    =\diag\vec{\beta}_{k_2l_2},\label{eq:directional_spectral_uncorrelatedness}
  \end{align}
  and by putting $\vec{\alpha}_{i_2j_2}=\sum_{k_2}\sum_{l_2}\uid{u_{i_2k_2}}{2}\wideconj{\uid{u_{j_2l_2}}{2}}\vec{\beta}_{k_2l_2}$, $\mat{U}_1$ diagonalizes $\Cov{\app{x}{\placeholder,i_2},\app{x}{\placeholder,j_2}}$ into $\diag\vec{\alpha}_{i_2j_2}$.
  Now condition 1) holds.
\end{proof}

Second, we show a proof of \cref{th:directional_stationarity_nscond}(A), which implies that of the part (B).
Regarding a matrix $\mat{H}_{s_1}$ in \cref{eq:G1_stationarity}, let $h_{s_1,i_2j_2}$ be the $\seqprn{i_2,j_2}$-th element of $\mat{H}_{s_1}$ and define a vector $\vec{h}_{i_2j_2}$ with $s_1$-th element $h_{s_1,i_2j_2}$.
Note that all the indices $i_2$, $j_2$, and $s_1$ start from zeros.
By putting $\vtilde{x}_{k_1i_2}=\sum_{i_1}x_{i_1i_2}\wideconj{\uid{u_{i_1k_1}}{1}}$ and $\vtilde{z}_{k_1i_2}=\sum_{i_1}z_{i_1i_2}\wideconj{\uid{u_{i_1k_1}}{1}}$, a filter in the equation is represented as
\begin{align}
  \vtilde{x}_{k_1i_2}
  =\sum_{s_1}\paren{\vuid{\lambda_{k_1}}{1}}^{s_1}\sum_{\sigma_2}\vtilde{z}_{k_1\sigma_2}h_{s_1,\sigma_2i_2}
  =\sum_{\sigma_2}\vtilde{h}_{k_1,\sigma_2i_2}\vtilde{z}_{k_1\sigma_2}\label{eq:G1_stationarity_spec}
\end{align}
where $\vtilde{h}_{k_1,\sigma_2i_2}$ is the $k_1$-th element of a vector $\vtilde{\vec{h}}_{\sigma_2i_2}=\mat{\Psi}_1\vec{h}_{\sigma_2i_2}$ for $k_1=0,\ldots,N_1-1$.

\begin{lemma}\label{lem:G1_stationary_spcov}
  When a signal $x$ on the graph $G_1\cprod G_2$ is $G_1$-stationary and represented as \cref{eq:G1_stationarity}, an equation
  \begin{align}
    \Cov{\vtilde{x}_{k_1i_2},\vtilde{x}_{l_1j_2}}
    =\delta_{k_1l_1}\sum_{\sigma_2}\tilde{h}_{k_1,\sigma_2i_2}\wideconj{\tilde{h}_{k_1,\sigma_2j_2}}
  \end{align}
  holds.
\end{lemma}

\begin{proof}
  Using an equation
  \begin{align}
    &\Cov{\vtilde{z}_{k_1i_2},\vtilde{z}_{l_1j_2}}
    =\sum_{i_1}\sum_{j_1}\wideconj{\uid{u_{i_1k_1}}{1}}\uid{u_{j_1l_1}}{1}\EV{z_{i_1i_2}z_{j_1j_2}}
    \\
    &=\sum_{i_1}\sum_{j_1}\wideconj{\uid{u_{i_1k_1}}{1}}\uid{u_{j_1l_1}}{1}\delta_{i_1j_1}\delta_{i_2j_2}
    =\delta_{k_1l_1}\delta_{i_2j_2},
  \end{align}
  covariance of $\vtilde{x}_{k_1k_2}$ in \cref{eq:G1_stationarity_spec} is given by
  \begin{align}
    \Cov{\vtilde{x}_{k_1i_2},\vtilde{x}_{l_1j_2}}
    &=\sum_{\sigma_2}\sum_{\tau_2}\tilde{h}_{k_1,\sigma_2i_2}\wideconj{\tilde{h}_{l_1,\tau_2j_2}}\EV{\tilde{z}_{k_1\sigma_2}\wideconj{\tilde{z}_{l_1\tau_2}}}
    \\
    &=\delta_{k_1l_1}\sum_{\sigma_2}\tilde{h}_{k_1,\sigma_2i_2}\wideconj{\tilde{h}_{k_1,\sigma_2j_2}}.
  \end{align}
\end{proof}

\begin{proof}[Proof of \cref{th:directional_stationarity_nscond}(A)]
  [from $G_1$-stationarity to simultaneous diagonalizability]
  Assuming the signal $x$ is a $G_1$-stationary, a covariance matrix $\Cov{\app{x}{\placeholder,i_2},\app{x}{\placeholder,j_2}}$ is diagonalizable as
  \begin{align}
    \adjoint{\mat{U}_1}\Cov{\app{x}{\placeholder,i_2},\app{x}{\placeholder,j_2}}\mat{U}_1
    &=\Cov{\vtilde{\vec{x}}_{\nid i_2},\vtilde{\vec{x}}_{\nid j_2}}
    \\
    &=\app{\diag}{\sum_{s_2}\tilde{\vec{h}}_{s_2i_2}\hprod\wideconj{\tilde{\vec{h}}_{s_2j_2}}}
  \end{align}
  due to \cref{lem:G1_stationary_spcov}.

  [from simultaneous diagonalizability to $G_1$-stationarity]
  For a zero-mean signal $x$ on the graph $G_1\cprod G_2$, suppose that a covariance matrix $\Cov{\app{x}{\placeholder,i_2},\app{x}{\placeholder,j_2}}$ is simultaneously diagonalizable with $\mat{L}_1$, i.e., represented as
  \begin{align}
    \Cov{\app{x}{\placeholder,i_2},\app{x}{\placeholder,j_2}}
    =\mat{U}_1\paren{\diag\vec{\gamma}_{i_2j_2}}\adjoint{\mat{U}_1}
  \end{align}
  with some vector $\vec{\gamma}_{i_2j_2}$ for any $i_2,j_2=0,\ldots,N_2-1$.
  Suppose the eigenvalues $\uid{\lambda_0}{1},\ldots,\uid{\lambda_{N_1-1}}{1}$ are distinct.
  Define a matrix $\mat{\Gamma}_{k_1}$ whose $\seqprn{i_2,j_2}$-th element is the $k_1$-th element of the vector $\vec{\gamma}_{i_2j_2}$ for $k_1=0,\ldots,N_1-1$, and a matrix $\vtilde{\mat{H}}_{k_1}$ satisfying a Cholesky decomposition $\adjoint{\tilde{\mat{H}}_{k_1}}\tilde{\mat{H}}_{k_1}^{}=\trsps{\mat{\Gamma}_{k_1}}$.
  The matrix $\vtilde{\mat{H}}_{k_1}$ exists because $\mat{\Gamma}_{k_1}=\Cov{\vtilde{\vec{x}}_{k_1\nid},\vtilde{\vec{x}}_{k_1\nid}}$ is symmetric and positive-semidefinite.
  Then the signal $x$ can be represented like \cref{eq:fgwise_stationarity} with $\mat{H}_{k_1}$ defined by $\vec{h}_{i_2j_2}=\mat{\Psi}_1^{-1}\vtilde{\vec{h}}_{i_2j_2}$, where $\vec{h}_{i_2j_2}$ and $\vtilde{\vec{h}}_{i_2j_2}$ are vectors whose $k_1$-th element is the $\seqprn{i_2,j_2}$-th element of $\mat{H}_{k_1}$ and $\vtilde{\mat{H}}_{k_1}$, respectively.
  The vector $\vec{h}_{i_2j_2}$ exists because the matrix $\mat{\Psi}_1$ is invertible due to the distinctness of the eigenvalues.
  Therefore, $x$ is $G_1$-stationary.
\end{proof}


\end{document}